\documentclass[a4paper,10pt,reqno]{amsart}

\usepackage{pstricks}

\usepackage{graphicx}
\usepackage{fancybox}
\usepackage[centertags]{amsmath}
\usepackage{amsfonts,amssymb,amsthm}
\usepackage{float,longtable}
\usepackage[centertags]{amsmath}
\usepackage{amsfonts,amssymb,amsthm}
\usepackage{float,longtable}
\usepackage{graphicx}
\usepackage{a4wide}
\usepackage{psfrag}

\usepackage[latin1]{inputenc}
\usepackage{newlfont}
 \theoremstyle{plain}
 \begingroup
 \newtheorem{thm}{Theorem}[section]
 \newtheorem{lem}[thm]{Lemma}

 \endgroup

 \theoremstyle{definition}
 \begingroup
 \newtheorem{defn}[thm]{Definition}
 
 \endgroup

 \theoremstyle{remark}
 \begingroup
 \newtheorem{rem}[thm]{Remark}
 \endgroup

 \numberwithin{equation}{section}

\title[Damage as $\Gamma$-limit of microfractures]
{Damage as $\Gamma$-limit of microfractures \\in anti-plane linearized elasticity}
\author[]{LUCIA SCARDIA}
\address[]{S.I.S.S.A., Via Beirut 2-4, 34014, Trieste, Italy}
\email[]{scardia@sissa.it}

\begin{document}
\maketitle
\begin{center}
\begin{minipage}{12cm}
\small{
\noindent {\bf Abstract.} A homogenization result is given for a material having 
brittle inclusions arranged in a periodic structure. According to the relation between the 
softness parameter and the size of the microstructure, three different limit models are deduced via 
$\Gamma$-convergence. In particular, damage is obtained as limit of periodically distributed
microfractures.

\vspace{15pt}

\vspace{6pt}
}
\end{minipage}
\end{center}

\bigskip

\section{Introduction}
The results contained in this paper describe the homogenization of a material composed by two
constituents which are distributed in a periodic way and which have a very different elastic 
behaviour. More precisely, we consider the case of an unbreakable elastic material presenting 
disjoint brittle inclusions arranged in a periodic way. 
In other words, we assume that cracks can appear and grow only in a prescribed disconnected region 
of the material, composed of a large number of small components with small toughness. 

In what follows, let $\Omega\subset\mathbb{R}^n$, with $n\geq 2$, be the region occupied by the material and let $\varepsilon>0$ be a small parameter. 
We introduce a structure on $\Omega$ whose periodicity cells $\varepsilon\, Q$ are the 
$\varepsilon$-homothetic of the unit square $Q:=(0,1)^n$. For any $0<\delta<1/2$ we denote 
with $Q_\delta \subset Q$ the concentric cube $(\delta, 1-\delta)^n$.
Let us focus on a single cell $\varepsilon\, Q$. 
We assume that cracks can appear only in a region contained in $\varepsilon\, Q_\delta$. Moreover, in 
order to deal with a quite general situation we allow the fragile part to have an $n$-dimensional
component and an $(n-1)$-dimensional one, which can be interpreted as a \textit{fissure} in the 
material. Hence, we consider an open set $E\subset Q_\delta$ and an $(n-1)$-dimensional set $F\subset Q_\delta$
and we require that the fracture in a single cell is contained in $\varepsilon\,E \cup \varepsilon\,F$.

A pictorial idea of the composition of the material is given by the following figure:

\vspace{-.2 cm}
\begin{figure}[htbp]
\begin{center}
\includegraphics[height=.3\textwidth]{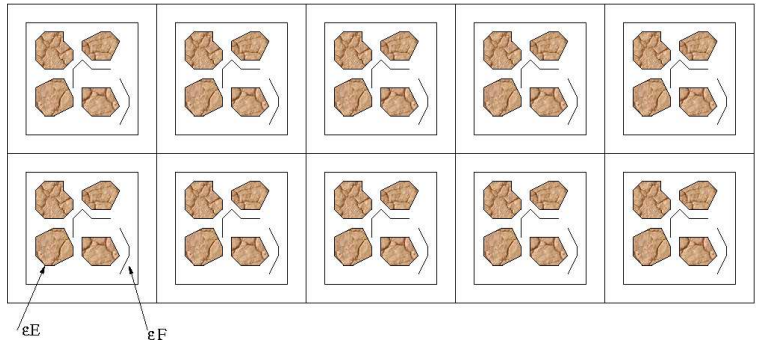}
\end{center}
\caption{Composite material}
\label{visc}
\end{figure}

To simplify the mathematical description of the model we consider only linearly elastic 
materials, and we restrict our analysis to the case of anti-plane shear. More precisely, 
we assume that the reference configuration is an infinite cylinder $\Omega \times \mathbb{R}$ and 
the displacement $v: \Omega\times \mathbb{R} \rightarrow \mathbb{R}^{n+1}$ has the special 
form $v(x,t):= (0,\dots,0,u(x))$ for every $(x,t)\in \Omega\times \mathbb{R}$, where 
$u: \Omega \rightarrow \mathbb{R}$. 

Since we are taking into account the possibility of creating cracks, displacements 
are allowed to have discontinuities. Therefore, the natural functional setting for the problem 
is the space of special functions with bounded variation. More precisely, we consider displacements 
$u\in SBV^2(\Omega)$, that is, we assume in addition that the approximate gradient $\nabla u$ is in 
$L^2$ and that the $(n-1)$-dimensional Hausdorff measure of the jump set $S_u$ is finite.

The elastic energy $\mathcal{F}^{\varepsilon}$ associated to a displacement 
$u\in SBV^2(\Omega)$ is defined as
\begin{equation*}
\mathcal{F}^{\varepsilon}(u) = \int_{\Omega} |\nabla  u|^2 dx + \int_{S_{u}}f_{\alpha_\varepsilon}
\Big(\frac{x}{\varepsilon}\Big) \,d \mathcal{H}^{n - 1}(x), 
\end{equation*}
where $f_{\alpha_\varepsilon}: \mathbb{R}^{n}
\rightarrow [0, + \infty]$ is a $Q$-periodic function defined as
\begin{equation*}
f_{\alpha_\varepsilon}(y) = 
\begin{cases}
\alpha_\varepsilon \quad & \hbox{in } E\cup F,\\
+ \infty \quad & \hbox{otherwise in } Q,
\end{cases}
\end{equation*}
and $\alpha_\varepsilon$ is a positive parameter depending on $\varepsilon$.

The volume term in the expression of $\mathcal{F}^{\varepsilon}$ represents the linearly elastic 
energy of the body, while the surface integral describes the energy needed in order to open a crack 
in a material with toughness $\alpha_\varepsilon$, according to Griffith's model of brittle fractures 
(see \cite{Grif}).
More precisely, the density $f_{\alpha_\varepsilon}$ acts as a \textit{weight} 
for the measure of the jump set $S_u$ of the displacement $u$. 
Indeed, the energy is finite only when $S_u$ lies in the fragile part of the material.
 
We are interested in the asymptotic behaviour of the sequence $\mathcal{F}^{\varepsilon}$
as $\varepsilon$ goes to zero, in the framework of $\Gamma$-convergence.

Heuristically, as $\varepsilon$ becomes smaller and smaller, the microscopic structure of the material 
becomes finer and finer, while, on the other hand, from a macroscopic point of view the 
behaviour of the composite tends to be simpler. 
So we expect the limit behaviour of the material to be described in terms of a different 
\textit{homogeneous} material, that captures the main features of the two original constituents. 

We consider the case in which $\delta$ is fixed and independent of $\varepsilon$, while $\alpha_\varepsilon$
converges to zero as $\varepsilon \rightarrow 0$. We show that the limit model depends on the
behaviour of the ratio $\frac{\alpha_\varepsilon}{\varepsilon}$ as $\varepsilon$ goes to zero.
However, it turns out that the different limiting models present a common feature: they describe an
unbreakable material.
This means that, even if at scale $\varepsilon$ many microscopic
cracks are present in the material, they are not equivalent 
in the limit model to a macroscopic crack, due to the fact that 
they are well separated from one another. 
Indeed, in the periodicity cell $\varepsilon\, Q$ the brittle inclusion 
$\varepsilon\, E \cup \varepsilon\, F$ is set at a distance $\varepsilon \delta$ 
from the boundary $\partial(\varepsilon\, Q)$, with  $\delta>0$ 
independent of $\varepsilon$. The size of the separation between different 
inclusions prevents the small cracks contained in the brittle region 
of the material from having the same asymptotic effect of a macroscopic fracture.

A different situation occurs when the parameter $\delta$ depends on $\varepsilon$ 
and converges to zero as $\varepsilon \rightarrow 0$.
This case has been partially solved in \cite{Bar}, assuming $\alpha_\varepsilon = 1$.

In this paper we show that three different limit models can arise, corresponding to the limit $\frac{\alpha_\varepsilon}{\varepsilon}$ 
being zero (subcritical case), finite (critical case) or $+ \infty$ (supercritical case).

In the subcritical case $\alpha_{\varepsilon}<< \varepsilon$, the limit functional turns out to be

\begin{equation*}
\mathcal{F}^0(u) = 
\begin{cases}
\displaystyle \int_{\Omega}f_{0}(D u)\,dx \quad &\hbox{in } H^1(\Omega),\\
+ \infty \quad  &\hbox{otherwise in } L^2(\Omega),
\end{cases}
\end{equation*}
where $f_{0}$ is a coercive quadratic form given by the cell formula
\begin{equation}\label{min}
f_{0}(\xi) = \min \bigg\{
\int_{Q\setminus (E \cup F)}|\,\xi + D w(y)|^2 dy : w\in H^1_{\#} (Q\setminus (E\cup F)) \bigg\},
\end{equation}
and $H^1_{\#}(Q\setminus (E\cup F)) $ denotes the space of $H^1(Q\setminus (E\cup F))$ functions with periodic boundary values on $\partial Q$.
Hence there exists a positive definite matrix 
$A_0 \in \mathbb{R}^{n \times n}$ with constant coefficients such that
$f_{0}(\xi) = A_{0}\xi \cdot \xi$ for every  $\xi \in \mathbb{R}^n$.
Notice that $\mathcal{F}^0$ represents the energy of a linearly elastic 
homogeneous anisotropic material. 
Moreover, since $w \equiv 0$ is a competitor for the minimum in (\ref{min}), the density $f_{0}$ satisfies
\begin{equation*}
A_0\xi \cdot \xi = f_{0}(\xi) \leq \big(1 - \mathcal{L}^n(E)\big)|\xi|^2 \leq |\xi|^2 \quad \hbox{for every } \xi \in \mathbb{R}^n,
\end{equation*}
and the second inequality is strict for $\xi \neq 0$. This means that ``$A_0 \lneqq Id$'' in the usual sense of 
quadratic forms. This is due to the fact that
in this regime, for the problem at fixed $\varepsilon$, displacements presenting discontinuities
are energetically convenient. Hence, although the limit energy $\mathcal{F}^0$ 
describes an unbreakable material, the possibility to create a high number of microfractures in
the approximating problems leads to a damaged limit material, that is, a material whose elastic properties 
are weaker than the original ones.

In the supercritical regime $\alpha_{\varepsilon} >> \varepsilon$ the limit model is described by the functional
\begin{equation*}
\mathcal{F}^{\infty}(u) = 
\begin{cases}
\displaystyle \int_{\Omega}|D u|^2\,dx \quad &\hbox{in } H^1(\Omega),\\
+ \infty \quad  &\hbox{otherwise in } L^2(\Omega).
\end{cases}
\end{equation*}
Hence, the (possible) presence of small cracks in the problems at scale $\varepsilon$ 
does not affect the elastic properties of the original material.
Indeed, in this regime the formation of microfractures is penalized by the energy, that is,
displacements presenting jumps are not energetically convenient.
Therefore the macroscopic result describes an undamaged material.

The critical regime corresponds to the case where $\alpha_\varepsilon$ is of the same
order as $\varepsilon$, so we can assume without loss of generality that $\alpha_\varepsilon = \varepsilon$.
The limit functional is
\begin{equation*}
\mathcal{F}^{hom}(u) = 
\begin{cases}
\displaystyle \int_{\Omega}f_{hom}(D u)\,dx \quad &\hbox{in } H^1(\Omega),\\
+ \infty \quad  &\hbox{otherwise in } L^2(\Omega),
\end{cases}
\end{equation*}
where the density $f_{hom}$ is given by the asymptotic cell formula
\begin{align}\label{density}
f_{hom}(\xi):= \lim_{t\rightarrow +\infty}\frac{1}{t^n}\,\inf\bigg\{\int_{(0,t)^n} 
|\xi + \nabla w|^2 d\,x + \mathcal{H}^{n-1}(S_w): w\in SBV^2_0\big((0,t)^n\big), 
S_w \subset \tilde{E}\cup\tilde{F} \bigg\}, 
\end{align}
and the sets $\tilde{E}$ and $\tilde{F}$ are defined as
\begin{equation*}
\tilde{E}:= E + \mathbb{Z}^n, \quad \tilde{F}:= F + \mathbb{Z}^n.
\end{equation*}
Notice that, since in this case the coefficient $\alpha_\varepsilon$ and the size $\varepsilon$ of the
microstructure have the same order, there is a competition between the bulk energy and the
surface term. Indeed they both contribute to the expression of the limit density. 

Moreover, the limit functional describes an intermediate model with respect to the
subcritical and the supercritical regimes. More precisely, the limit density satisfies
\begin{equation}\label{ineqqq}
f_0(\xi) \lneqq f_{hom}(\xi) \leq \min \big\{|\xi|^2, f_0(\xi) + c(E)\big\},
\end{equation}
for every $\xi \in \mathbb{R}^n \setminus \{0\}$, where $c(E)$ is the $(n-1)$-dimensional 
measure of $\partial E$ (see Lemma \ref{notquad}).

Notice that (\ref{ineqqq}) entails that for $|\xi|$ large enough $f_{hom}(\xi) \lneqq |\xi|^2$.  
Therefore, the limit functional describes a damaged material. 
Using estimate (\ref{ineqqq}) it is also possible to show that the limit density $f_{hom}$ is
\textit{not} two-homogeneous, and hence it is not a quadratic form (see again Lemma \ref{notquad}). 

The plan of the paper is the following. 
In Section 2 we define the energy functional and we describe the mathematical setting 
of the problem. Sections 3-5 are devoted to the asymptotic analysis of the energy in the 
various regimes and to the description of the limit functionals in the subcritical, critical 
and supercritical cases. 
In the last Section we present, in the two-dimensional case, an alternative and direct proof 
of the main result of Section 4, in the regime $\alpha_\varepsilon >> \varepsilon$.

\section{Preliminaries and formulation of the problem}
\noindent
Let us give some definitions and results that will be widely used throughout the paper.

In order to make precise the mathematical setting of this problem, we need to recall some properties
of rectifiable sets and of the space $SBV$ of special functions with bounded variation.
We refer the reader to \cite{AFP} for a complete treatment of these subjects.

A set $\Gamma\subset \mathbb{R}^n$ is rectifiable if there exist $N_0 \subset \Gamma$ with 
$\mathcal{H}^{n-1}(N_0) = 0$, and a sequence $(M_i)_{i \in \mathbb{N}}$ of $C^1$-submanifolds of 
$\mathbb{R}^n$ such that
\begin{equation*}
 \Gamma \setminus N_0 \subset \bigcup_{i\in\mathbb{N}} M_i.
\end{equation*}
For every $x\in\Gamma\setminus N_0$ we define the normal to $\Gamma$ at $x$ as $\nu_{M_i}(x)$. 
It turns out that the normal is well defined (up to the sign) for $\mathcal{H}^{n-1}$-a.e. 
$x\in\Gamma$.

Let $U\subset \mathbb{R}^n$ be an open bounded set with Lipschitz boundary. We define $SBV(U)$ as 
the set of functions $u\in L^1(U)$ such that the distributional derivative $Du$ is a Radon measure which, 
for every open set $A\subset U$, can be represented as
\begin{equation*}
 Du(A) = \int_{A} \nabla u\, dx + \int_{A\cap S_u} [u](x)\,\nu_u(x)\, d\mathcal{H}^{n-1}(x),
\end{equation*}
where $\nabla u$ is the approximate differential of $u$, $S_u$ is the set of jump of $u$ (which is a 
rectifiable set), $\nu_u(x)$ is the normal to $S_u$ at $x$, and $[u](x)$ is the jump of $u$ at $x$.

\noindent
For every $p \in ]1, +\infty[$ we set
\begin{equation*}
 SBV^p(U) = \big\{u\in SBV(U): \nabla u \in L^p (U;\mathbb{R}^n), \mathcal{H}^{n-1}(S_u) < +\infty \big\}.
\end{equation*}
If $u\in SBV(U)$ and $\Gamma\subset U$ is rectifiable and oriented by a normal vector field $\nu$, then we 
can define the traces $u^+$ and $u^-$ of $u\in SBV(U)$ on $\Gamma$ which are characterized by the relations
\begin{equation*}
\lim_{r\rightarrow 0}\, \frac{1}{r^n} \int_{\Omega \cap B^{\pm}_r(x)} |u(y) - u^{\pm}(x)|\,dy = 0 \quad 
\hbox{for } \mathcal{H}^{n-1}- \hbox{a.e. } x\in \Gamma,
\end{equation*}
where $B^{\pm}_r(x) := \{y\in B_r(x): (y - x)\cdot \nu \gtrless 0 \}$.

\noindent
A set $E\subset U$ has finite perimeter in $U$ if the characteristic function $\chi_E$ belongs to $SBV(U)$. 
We denote by $\partial^*E$ the set of jumps of $\chi_E$ and by $P(E,U)$ the total variation of the measure 
$D\chi_E$, that is, the perimeter of $E$ in $U$.

Finally, if $E\subset U$, we denote with $E(\sigma)$ the set of points of density $\sigma \in [0,1]$ for $E$, i.e., 
\begin{equation*}
 E(\sigma):= \big\{ x\in U: \lim_{r\rightarrow 0} \mathcal{L}^n(E \cap B_r(x))/ \mathcal{L}^n(B_r(x)) = \sigma \big\}.
\end{equation*}
\vspace{.2cm}

Let us come to the formulation of the problem.

\noindent
Let $n \geq 2$ and let $\Omega\subset \mathbb{R}^n$ be a bounded open set. 
In the following we will denote by $Q$ the unit cube $(0,1)^n$ and by $Q_{\varrho}$ the inner  
cube $(\varrho, 1- \varrho)^n$, for some $\varrho \in (0,1)$.

\noindent
Let $\delta > 0$ and $E, F\subset Q_\delta$ be defined in the following way:
\begin{itemize}
\item $E$ is a finite union of disjoint sets given by the closure of domains with Lipschitz 
boundary;

\vspace{.1cm}
\item $F$ is a finite union of disjoint closed $(n - 1)$-dimensional smooth manifolds.
\end{itemize}
Assume also that $E$ and $F$ are disjoint.

For every $\varepsilon > 0$, let us consider the periodic structure in $\mathbb{R}^n$ generated 
by an $\varepsilon$-homothetic of the basic cell $Q$.  

The starting point of the problem is the energy associated
to a function $u\in SBV^2(\Omega)$, that is
\begin{equation*}
\mathcal{F}^{\varepsilon}(u) = \int_{\Omega} |\nabla  u|^2 dx + \int_{S_{u}}f_{\alpha}
\Big(\frac{x}{\varepsilon}\Big) \,d \mathcal{H}^{n - 1}(x), 
\end{equation*}
where $f_{\alpha}: \mathbb{R}^{n}
\rightarrow [0, + \infty]$ is a $Q$-periodic function defined as
\begin{equation*}
f_{\alpha}(y) = 
\begin{cases}
\alpha \quad & \hbox{in } E\cup F,\\
+ \infty \quad & \hbox{otherwise in } Q,
\end{cases}
\end{equation*}
and $\alpha$ is a positive parameter.
Clearly, being $f_{\alpha}$ $Q$-periodic, the function
\begin{equation*}
x \mapsto f_{\alpha} \Big(\frac{x}{\varepsilon}\Big)
\end{equation*}
turns out to be $\varepsilon\,Q$-periodic. For notational brevity we will use the superscript $\varepsilon$ to 
denote the $\varepsilon$-homothetic of any domain. 
In particular, $Q^{\varepsilon} := \varepsilon\, Q$.

Let us write the domain $\Omega$ as union of cubes of side $\varepsilon$:
\begin{equation*}
\Omega = \bigg(\bigcup_{h\in \mathbb{Z}^n_{\varepsilon}} (Q + h)^{\varepsilon}\bigg)\cup R(\varepsilon), 
\end{equation*}
where $\mathbb{Z}_{\varepsilon}^n$ is the set of integer vectors $h\in \mathbb{Z}^n$ such that 
$(Q + h)^{\varepsilon}\subset \Omega$ and $R(\varepsilon)$ is the remaining part of $\Omega$. 
Let $N(\varepsilon)$ be the cardinality of the set $\mathbb{Z}^n_{\varepsilon}$; notice that 
$N(\varepsilon)$ is of order $1/\varepsilon^n$.

We denote by $\{Q^\varepsilon_k\}_{k=1,\dots,N(\varepsilon)}$ an enumeration of the family of cubes 
$(Q + h)^{\varepsilon}$ covering $\Omega$, so that we can rewrite $\Omega$ as 
\begin{equation}\label{domom}
\Omega = \Bigg(\bigcup_{k=1}^{N(\varepsilon)} Q^{\varepsilon}_k\Bigg)\cup R(\varepsilon).
\end{equation}

Let $E^{\varepsilon}_k, F^{\varepsilon}_k\subset\subset Q^{\varepsilon}_k$ be defined in the same way. 
Finally, we set  
\begin{equation}\label{subd}
\tilde{E}^{\varepsilon}:= \bigg(\bigcup_{k=1}^{N(\varepsilon)} E^{\varepsilon}_k\bigg)\cup R_{E}(\varepsilon), \quad
\tilde{F}^{\varepsilon}:=  \bigg(\bigcup_{k=1}^{N(\varepsilon)} F^{\varepsilon}_k\bigg)\cup R_{F}(\varepsilon),
\end{equation}
where $R_{E}(\varepsilon)$ and $R_{F}(\varepsilon)$ are the remaining parts of $\Omega\cap (E + \mathbb{Z}^n)^\varepsilon$ and of $\Omega \cap (F + \mathbb{Z}^n)^\varepsilon$, respectively. 

\vspace{.3cm}
We are interested in the case in which $\delta$ is fixed and independent of $\varepsilon$, while $\alpha = \alpha_\varepsilon$ depends on $\varepsilon$ and goes to zero as $\varepsilon \rightarrow 0$.

We will study three different cases, i.e.,
\begin{equation*}
\hspace{-3cm}
\begin{array}{lll}
\vspace{.15 cm}
\hbox{{1. Subcritical regime}}\quad &\dfrac{\alpha_\varepsilon}{\varepsilon} \rightarrow 0 \quad &\hbox{as } \varepsilon \rightarrow 0,\\
\vspace{.15 cm}
\hbox{{2. Supercritical regime}}\quad &\dfrac{\alpha_\varepsilon}{\varepsilon} \rightarrow +\infty \quad &\hbox{as } \varepsilon \rightarrow 0,\\
\hbox{{3. Critical regime}}\quad &\dfrac{\alpha_\varepsilon}{\varepsilon} \rightarrow c \in (0,+\infty) \quad &\hbox{as } \varepsilon \rightarrow 0.
\end{array}
\end{equation*}

\noindent
Before starting the analysis of the different cases we have just described, we state a fundamental result
that will be often used in the following. For the proof we refer to \cite{ACDP}.

\begin{thm}[Existence of an extension operator]\label{extop}
Let $E$ be a periodic, connected, open subset of $\mathbb{R}^n$, with Lipschitz boundary, let 
$\varepsilon > 0$, and set $E^\varepsilon:= \varepsilon\,E$. Given a bounded open set 
$\Omega \subset \mathbb{R}^n$, there exist a linear and continuous extension operator 
$T^\varepsilon : H^1(\Omega\cap E^\varepsilon) \rightarrow H^1_{loc}(\Omega)$ and three
constants $k_0, k_1, k_2 > 0$ depending on $E$ and $n$, but not on $\varepsilon$ and $\Omega$, such that 
\begin{align*}
 T^\varepsilon u =\,& u \,\, \mbox{a.e. in }\Omega\cap E^\varepsilon,\\
\int_{\Omega(\varepsilon k_0)}|T^\varepsilon u|^2 dx \leq& \,k_1 \int_{\Omega\cap E^\varepsilon} |u|^2 dx,\\
\int_{\Omega(\varepsilon k_0)}|D(T^\varepsilon u)|^2 dx \leq&\, k_2 \int_{\Omega\cap E^\varepsilon} |Du|^2 dx,
\end{align*}
for every $u \in H^1(\Omega\cap E^\varepsilon)$. Here we used the notation 
$\Omega(\varepsilon k_0):= \{ x \in \Omega: \mbox{dist}(x,\partial \Omega)>\varepsilon k_0 \}$.
\end{thm}

\begin{rem}
Theorem \ref{extop} applies to a very large class of domains $E$. In particular, it covers the 
case in which $E$ is obtained by removing from the periodicity cell $Q:= (0,1)^n$ a set $B$ 
with Lipschitz boundary such that $\mbox{dist}(B, \partial Q) > 0$, and repeating this structure
by periodicity (see also \cite{Khru}).  
\end{rem}

\section{Subcritical regime: very brittle inclusions} 
\noindent
In this section we assume $\alpha_\varepsilon << \varepsilon$ in the expression of the energy 
$\mathcal{F}^\varepsilon$.

We define the functional $\mathcal{F}^0: L^{2}(\Omega) \rightarrow [0, + \infty]$ as

\begin{equation}\label{defG0}
\mathcal{F}^0(u) =
\begin{cases}
\displaystyle \int_{\Omega}f_{0}(D u)\,dx \quad &\hbox{if } u\in H^1(\Omega),\\
+ \infty \quad  &\hbox{otherwise in } L^2(\Omega),
\end{cases}
\end{equation}
where $f_{0}$ solves the cell problem
\begin{equation}\label{cell0}
f_{0}(\xi) = \min \bigg\{
\int_{Q\setminus (E\cup F)}|\xi + D w(y)|^2 dy : w\in H^1_{\#} (Q\setminus (E\cup F)) \bigg\}.
\end{equation}
The functional $\mathcal{F}^0$ will turn out to be the $\Gamma$-limit of the sequence $(\mathcal{F}^\varepsilon)$ 
in this case, that is for $\alpha_\varepsilon << \varepsilon$.

It is convenient to introduce the auxiliary functionals $\mathcal{G}^{\varepsilon}: L^2(\Omega) \rightarrow [0,+\infty]$ 
defined by 
\begin{equation}\label{defGe}
\mathcal{G}^{\varepsilon}(v) = 
\begin{cases}
\displaystyle \int_{\Omega} a\Big(\frac{x}{\varepsilon}\Big)|\nabla v|^2 dx & 
\hbox{if } v\in H^1(\Omega\setminus\tilde{F}^\varepsilon),\\
+\infty & \hbox{otherwise in } L^2(\Omega),
\end{cases}
\end{equation}
where $a$ is a $Q$-periodic function given by
\begin{equation*}
a(y) =
\begin{cases}
0 \quad & \hbox{in } E,\\
1 \quad & \hbox{in } Q\setminus E.
\end{cases}
\end{equation*}

\noindent
As a preliminary result, we show that $\mathcal{G}^{\varepsilon}$ $\Gamma$-converges to $\mathcal{F}^0$ 
with respect to the strong topology of $L^2_{loc}$.


\begin{thm}\label{GammaG}
The sequence of functionals $(\mathcal{G}^{\varepsilon})$ $\Gamma$-converges to $\mathcal{F}^0$ 
with respect to the strong topology of $L^2_{loc}$.
\end{thm}
\begin{proof}
Let $\eta > 0$ and let $F_{\eta}$ be a neighbourhood of $F$ with Lipschitz boundary
such that $\hbox{dist}(F_{\eta},F)\leq \eta$ and $\hbox{dist}(F_{\eta},E) > 0$.
Now we define the functionals $\mathcal{G}_{\eta}^{\varepsilon}: L^2(\Omega) 
\rightarrow [0, + \infty]$ as

\begin{equation}\label{defGee}
\mathcal{G}_{\eta}^{\varepsilon}(v) = 
\begin{cases}
\displaystyle \int_{\Omega} a_\eta\Big(\frac{x}{\varepsilon}\Big)|\nabla v|^2 dx & 
\hbox{if } v\in H^1(\Omega),
\\
+\infty & \hbox{otherwise in } L^2(\Omega),
\end{cases}
\end{equation}
where $a_\eta$ is a $Q$-periodic function given by
\begin{equation*}
a_\eta(y) = 
\begin{cases}
0 \quad & \hbox{if } y \in E\cup F_\eta,\\
1 \quad & \hbox{otherwise in } Q.
\end{cases}
\end{equation*}
\noindent
From the standard theory for non-coercive convex homogenization (see e.g. \cite{Att} and \cite{BDf}), we know that 
\begin{equation}\label{Gamconv} 
\Gamma (L^2_{loc}) \mbox{-} \lim_{\varepsilon \rightarrow 0}\mathcal{G}_{\eta}^{\varepsilon} = \mathcal{G}_{\eta},
\end{equation}
where the functional $\mathcal{G}_{\eta} : L^2(\Omega) \rightarrow [0, + \infty]$ is defined as
\begin{equation*}
\mathcal{G}_{\eta}(v) = 
\begin{cases}
\displaystyle \int_{\Omega}f_{\eta}(D v)\,dx \quad &\hbox{if } v\in H^1(\Omega),\\
+ \infty \quad  &\hbox{otherwise in } L^2(\Omega),
\end{cases}
\end{equation*}
and $f_{\eta}$ solves for every $\xi \in \mathbb{R}^n$ the cell problem
\begin{align*}
f_{\eta}(\xi) =& \min \bigg\{
\int_{Q\setminus (E\cup F_{\eta})}|\xi + D w(y)|^2 dy : w\in H^1_{\#} (Q\setminus (E\cup F_\eta))\bigg\}\\
=& \min \bigg\{
\int_{Q\setminus (E\cup F_{\eta})}|\xi + D w(y)|^2 dy : w\in H^1_{\#} (Q)\bigg\}.
\end{align*}
Notice that the last equality is due to classical extension theorems (see, for instance, \cite{Ada}).

\textit{Comparison between $\mathcal{G}^{\varepsilon}$ and $\mathcal{G}_{\eta}^{\varepsilon}$}. 
Let $v^\varepsilon$ be a sequence having equibounded energies $\mathcal{G}^{\varepsilon}$ and such that 
$v^\varepsilon$ converges strongly to some $v$ in $L^2_{loc}$. Then we claim that $v \in H^1(\Omega)$  
and that
\begin{equation}\label{confr2}
\liminf_{\varepsilon \rightarrow 0}\mathcal{G}^{\varepsilon}(v^{\varepsilon}) \geq \mathcal{G}_{\eta}(v). 
\end{equation}
By the fact that $\mathcal{G}^{\varepsilon}(v^{\varepsilon})$ are bounded we deduce in particular that the $H^1(\Omega\setminus (\tilde{E}^{\varepsilon}\cup \tilde{F}_\eta^{\varepsilon}))$ 
norm of $v^\varepsilon$ is equibounded. 

Therefore, Theorem \ref{extop} ensures that for every $\varepsilon > 0$ 
there exists an extension of $v^\varepsilon$, that is a function $\tilde{v}_{\eta}^{\varepsilon} 
\in H^1_{loc}(\Omega)$ such that 
\begin{equation}\label{ext}
\tilde{v}_{\eta}^{\varepsilon} = v^\varepsilon \quad \hbox{in } \Omega\setminus (\tilde{E}^{\varepsilon}\cup \tilde{F}_\eta^{\varepsilon}),
\end{equation}
with the property that for every open Lipschitz set $\Omega'\subset \Omega$ such that 
$\mbox{dist}(\Omega',\partial \Omega) > k_0\varepsilon$, the $H^1(\Omega')$-norm of 
$\tilde{v}_{\eta}^{\varepsilon}$ is equibounded.
Hence there exists a function $v^*\in H^1(\Omega')$ such that
\begin{equation*}
\tilde{v}_{\eta}^{\varepsilon} \rightharpoonup v^* \quad \hbox{weakly in } H^1(\Omega') \quad \mbox{as } \varepsilon \rightarrow 0,
\end{equation*}
and strongly in $L^2(\Omega')$. 
If we now consider an invading sequence of smooth open subsets of $\Omega$, by a diagonal process 
we can extract a subsequence of $(\tilde{v}_{\eta}^{\varepsilon})$ (still denoted by $\tilde{v}_{\eta}^{\varepsilon}$) 
that converges to a function $v^* \in H^1_{loc}(\Omega)$, strongly in $L^2_{loc}(\Omega)$ and 
weakly in $H^1_{loc}(\Omega)$. It is easy to show that 
$v = v^*$ a.e. in $\Omega$. Indeed, using the relation (\ref{ext}) we have that for every open set 
$\Omega'\subset \subset \Omega$
\begin{equation*}
\int_{\Omega'\setminus (\tilde{E}^{\varepsilon}\cup \tilde{F}_\eta^{\varepsilon})}|v - v^*|^2 dx \leq 
\int_{\Omega'\setminus (\tilde{E}^{\varepsilon}\cup \tilde{F}_\eta^{\varepsilon})}|v - v^\varepsilon|^2 dx + 
\int_{\Omega'\setminus (\tilde{E}^{\varepsilon}\cup \tilde{F}_\eta^{\varepsilon})}|\tilde{v}_{\eta}^{\varepsilon} - v^*|^2 dx,
\end{equation*}
from which, by taking the limit as $\varepsilon \rightarrow 0$ we get
\begin{equation*}
\mathcal{L}^n(Q\setminus (E\cup F_{\eta}))\int_{\Omega'}|v - v^*|^2 dx \leq 0.
\end{equation*}
Since this holds for every $\Omega'\subset \subset \Omega$, we obtain $v\in H^1(\Omega)$.

\noindent
Moreover, the extension we have built allows us to write the estimate
\begin{equation}\label{confr}
\mathcal{G}^{\varepsilon}(v^\varepsilon) \geq \mathcal{G}_{\eta}^{\varepsilon}(\tilde{v}_{\eta}^{\varepsilon}),
\end{equation}
and in virtue of the result (\ref{Gamconv}) we get (\ref{confr2}).
It remains to show that on $H^1(\Omega)$ the $\Gamma$-limit of the sequence $\big(\mathcal{G}^{\varepsilon}\big)$ 
is given by $\mathcal{F}^{0}$, where $\mathcal{F}^{0}$ is defined by (\ref{defG0}) and (\ref{cell0}). 

\textit{Liminf inequality}.\\
Let $v\in H^1(\Omega)$ and let $(v^\varepsilon)$ be a sequence having equibounded energy $\mathcal{G}^{\varepsilon}$, such 
that $v^\varepsilon$ converges to $v$ strongly in $L^2$. Then (\ref{confr2}) holds for every $\eta > 0$.

Since $f_{\eta}$ converges increasingly to $f_0$, then $f_0 = \sup_{\eta} f_{\eta} = \lim_{\eta \rightarrow 0} f_{\eta}$. Hence 
$$\sup_{\eta}  \mathcal{G}_{\eta} = \mathcal{F}^{0},$$
and then from (\ref{confr2}) we get the bound
\begin{equation*}
\liminf_{\varepsilon \rightarrow 0}\mathcal{G}^{\varepsilon}(v^{\varepsilon}) \geq \mathcal{F}^{0}(v).
\end{equation*}

\textit{Limsup inequality}.
Let $\xi \in \mathbb{R}^n$ and let us define $v_{\xi}(x):= \xi\cdot x$. Let $w$ be the solution of the 
minimum problem defining $f_0(\xi)$, that is, $w \in H^1_{\#}(Q\setminus(E\cup F))$, and
\begin{equation*}
f_0(\xi) = \int_{Q\setminus (E\cup F)} |\xi + D w|^2 dx.
\end{equation*}
Let $\tilde{w}$ be the periodic extension of $w$ to $\mathbb{R}^n$ and let us define the sequence 
$v^\varepsilon := v_{\xi} + \varepsilon\,\tilde{w} \Big(\dfrac{x}{\varepsilon}\Big)$; clearly it 
converges to $v_{\xi}$ strongly in $L^2$. Moreover
\begin{align*}
\mathcal{G}^{\varepsilon}(v^\varepsilon) =&\, \int_{\Omega} a\Big(\frac{x}{\varepsilon}\Big)|\nabla v^\varepsilon|^2 dx 
= \,\varepsilon^n \int_{\Omega/\varepsilon} a(x)\,|\xi + \nabla \tilde{w}|^2 dx 
=\,\mathcal{L}^n(\Omega) \int_{Q} a(x)\,|\xi + \nabla w|^2 dx + o(\varepsilon)\\
=&\,\, \mathcal{L}^n(\Omega) \int_{Q\setminus E} |\xi + \nabla w|^2 dx + o(\varepsilon) = 
\mathcal{L}^n(\Omega) \,f_0(\xi) + o(\varepsilon)= \mathcal{F}^{0}(v_\xi) + o(\varepsilon),
\end{align*}
where $o(\varepsilon)$ is a small error that disappears when $\varepsilon \rightarrow 0$ and which is 
due to the fact that in general $\Omega/\varepsilon$ is not given by an exact number of unit cubes. 

We have therefore proved the existence of a recovery sequence for affine functions. We can extend the 
result to piecewise affine continuous functions, thanks to the local character of $\mathcal{G}^{\varepsilon}$.  
Then, using the density in $H^1(\Omega)$ of the piecewise affine continuous functions and the continuity
of $\mathcal{F}^{0}$ on $H^1(\Omega)$, we get the claim in the general case.
\end{proof}


\begin{rem}
From the previous result we deduce immediately that $f_0$ is a quadratic form, being $\mathcal{F}^0$ the 
$\Gamma$-limit of the quadratic forms $\mathcal{G}^\varepsilon$. Hence there exists a matrix 
$A_0 \in \mathbb{R}^{n \times n}$ with constant coefficients such that
\begin{equation}\label{defA0}
f_{0}(\xi) = A_{0}\xi \cdot \xi \quad \hbox{for every } \xi \in \mathbb{R}^n.
\end{equation}
\end{rem}

Now we can prove the $\Gamma$-convergence result for the sequence $\mathcal{F}^{\varepsilon}$.


\begin{thm}[Bound from below]\label{liminf1}
Let $u\in L^2(\Omega)$ and let $(u^\varepsilon)$ be a sequence with equibounded energy $\mathcal{F}^\varepsilon$
such that $u^\varepsilon\rightarrow u$ strongly in $L^2$. Then $u\in H^1(\Omega)$ and 
\begin{equation}\label{impbound}
\liminf_{\varepsilon\rightarrow 0}\mathcal{F}^{\varepsilon}(u^{\varepsilon})
\geq \mathcal{F}^{0}(u).
\end{equation}
\end{thm}
\begin{proof}
Let $u\in L^{2}(\Omega)$ and  let $(u^{\varepsilon})$ be a sequence converging to $u$
strongly in $L^2(\Omega)$ and such that $\mathcal{F}^{\varepsilon}(u^{\varepsilon}) \leq c < + \infty$.
From the definition of the functional this implies in particular that the $H^{1}(\Omega \setminus (\tilde{E}^{\varepsilon}\cup \tilde{F}^{\varepsilon}))$ norm of $(u^{\varepsilon})$ is equibounded.
By Theorem \ref{extop} it is possible to extend every $u^{\varepsilon}$ to a new function 
$\tilde{u}^{\varepsilon} \in H^1_{loc}(\Omega \setminus \tilde{F}^{\varepsilon})$ 
in such a way that for every open Lipschitz set $\Omega'\subset \Omega$ the $H^1$-norm of 
$\tilde{u}^{\varepsilon}$ in $\Omega'\setminus \tilde{F}^{\varepsilon}$ is equibounded.

We claim that $\tilde{u}^{\varepsilon} \rightarrow u$ strongly in $L^2(\Omega')$.
As first step, fix $\eta > 0$ and define for every $\varepsilon > 0$ an extension 
$\tilde{u}^{\varepsilon}_{\eta} \in H^1_{loc}(\Omega')$ of $(\tilde{u}^{\varepsilon})_{|\Omega' \setminus\tilde{F}^{\varepsilon}_{\eta}}$, where $\tilde{F}^{\varepsilon}_{\eta}$ denotes an 
$\eta$-neighborhood of $\tilde{F}^{\varepsilon}$ defined in the usual way. 
As in Theorem \ref{GammaG} it turns out that $\tilde{u}^{\varepsilon}_{\eta} \rightharpoonup u$ 
weakly in $H^1_{loc}(\Omega')$ and that $u \in H^1(\Omega')$. Moreover, 
\begin{align}\label{convconv}
\int_{\Omega'}|\tilde{u}^{\varepsilon} - u|^2 dx =&\, \int_{\Omega'\setminus \tilde{E}^\varepsilon}  |\tilde{u}^{\varepsilon} - u|^2 dx + \int_{\tilde{E}^\varepsilon}  |\tilde{u}^{\varepsilon} - u|^2 dx \nonumber\\
=&\, \int_{\Omega'\setminus \tilde{E}^\varepsilon}  |u^{\varepsilon} - u|^2 dx + \int_{\tilde{E}^\varepsilon}  |\tilde{u}^{\varepsilon}_\eta - u|^2 dx \nonumber\\
\leq&\, \int_{\Omega'}  |u^{\varepsilon} - u|^2 dx + \int_{\Omega'}  |\tilde{u}^{\varepsilon}_\eta - u|^2 dx, 
\end{align}
and since the right-hand side in (\ref{convconv}) converges to zero as $\varepsilon \rightarrow 0$, we can conclude that 
\begin{equation*}
\tilde{u}^{\varepsilon} \rightarrow u \quad \hbox{strongly in } L^2(\Omega'). 
\end{equation*}
Since this holds for every $\Omega'\subset \Omega$, we have that the convergence is indeed 
strong in $L^2_{loc}(\Omega)$ and that $u \in H^1(\Omega)$.
Using the sequence $\tilde{u}^{\varepsilon}$ we can write 
\begin{equation}\label{boundfb}
\mathcal{F}^{\varepsilon}(u^{\varepsilon}) \geq \mathcal{G}^{\varepsilon}(\tilde{u}^{\varepsilon}),
\end{equation}
where the functional $\mathcal{G}^{\varepsilon}$ is defined as in (\ref{defGe}).
Hence by Theorem \ref{GammaG} we obtain (\ref{impbound}).
\end{proof}

\begin{rem}\label{remark}
We underline that the bound (\ref{impbound}) holds true independently of the rate of convergence of  
$\alpha_\varepsilon$ and implies in particular that the $\Gamma$-limit of $\mathcal{F}^{\varepsilon}$ is
finite only in $H^1(\Omega)$.
\end{rem}


\begin{thm}[Bound from above]\label{bfab}
For every  $u\in H^1(\Omega)$ there exists a sequence $(u^\varepsilon)\subset
SBV^2(\Omega)$, with $S_{u}\subset \tilde{E}^\varepsilon \cup \tilde{F}^\varepsilon$, such that 
\begin{align*}
&(i) \quad u^\varepsilon \rightarrow u \quad \hbox{strongly in } L^2_{loc}(\Omega), \\
&(ii) \quad \lim_{\varepsilon\rightarrow 0}\mathcal{F}^{\varepsilon}(u^{\varepsilon}) = \mathcal{F}^{0}(u).
\end{align*}
\end{thm}

\begin{proof}
Let $u \in H^1(\Omega)$. The $\Gamma$-convergence result in Theorem \ref{GammaG} guarantees the 
existence of a sequence $(v^\varepsilon)\subset L^2(\Omega)$ such that
\begin{equation*}
\begin{cases}
v^\varepsilon \rightarrow u \quad \hbox{strongly in } L^2_{loc}(\Omega),\\
\mathcal{G}^{\varepsilon}(v^\varepsilon) \rightarrow \mathcal{F}^0(u).
\end{cases}
\end{equation*}
A recovery sequence for $\mathcal{F}^{\varepsilon}$ will be constructed by modifying 
properly $(v^\varepsilon)$.

\noindent
Notice that, by the definition of $\mathcal{G}^{\varepsilon}$, it turns out that
the $H^1(\Omega\setminus (\tilde{E}^{\varepsilon} \cup \tilde{F}^{\varepsilon}))$ norm of 
$v^\varepsilon$ is equibounded. We split the proof into three steps.

\textit{First step}. There exists a sequence $(\tilde{v}^\varepsilon)\subset H^1_{loc}(\Omega\setminus \tilde{F}^{\varepsilon})$ 
such that  
\begin{align*}
&(1)\quad \tilde{v}^\varepsilon = v^\varepsilon \quad \hbox{a.e. in } \Omega\setminus (\tilde{E}^{\varepsilon}\cup\tilde{F}^{\varepsilon}) ,\\
&(2)\quad ||\,\tilde{v}^\varepsilon ||_{H^1(\Omega'\setminus \tilde{F}^{\varepsilon})} \leq 
c\,||\,v^\varepsilon||_{H^1(\Omega\setminus (\tilde{E}^{\varepsilon} \cup \tilde{F}^{\varepsilon}))},
\end{align*}
for every open Lipschitz set $\Omega'\subset \Omega$ such that $\mbox{dist}(\Omega',\partial \Omega) > k_0\varepsilon$,
where the constant $c$ is independent of $\varepsilon$. This can be done exactly as in Theorem \ref{liminf1}. 

\textit{Second step}. The sequence $(\tilde{v}^\varepsilon)\subset H^1_{loc}(\Omega\setminus \tilde{F}^{\varepsilon})$ 
of the previous step is still a recovery sequence for $\mathcal{G}^{\varepsilon}$, i.e.,
\begin{align*}
&(3)\quad \tilde{v}^\varepsilon \rightarrow u \quad \hbox{strongly in } L^2_{loc}(\Omega),\\
&(4)\quad \mathcal{G}^{\varepsilon}(\tilde{v}^\varepsilon) \rightarrow \mathcal{F}^0(u).
\end{align*} 
Property (3) can be proved as in Theorem \ref{liminf1} while condition (4) follows immediately, 
since $\mathcal{G}^{\varepsilon}$ depends only on the behaviour of its argument in 
$\Omega\setminus \tilde{E}^\varepsilon$ and $v^\varepsilon$ and $\tilde{v}^\varepsilon$ agree on that set.

\textit{Third step}. There exists a sequence $(u^\varepsilon)\subset
SBV^2(\Omega)$ with $S_{u^\varepsilon}\subset \tilde{E}^\varepsilon\cup \tilde{F}^\varepsilon$ 
such that $(i)$ and $(ii)$ are satisfied.
Define 
\begin{equation*}
u^{\varepsilon}(x) := 
\begin{cases}
\tilde{v}^\varepsilon(x)\quad & \hbox{if } x\in \Omega\setminus \tilde{E}^{\varepsilon},\\
\tilde{v}^\varepsilon_k \quad & \hbox{if } x\in E_k^{\varepsilon},
\end{cases}
\end{equation*}
where $\tilde{v}^\varepsilon_k$ is the mean value of $\tilde{v}^\varepsilon$ over 
$E_k^{\varepsilon}$, for $k = 1,\dots,N(\varepsilon)$.
Then, for every $\Omega'\subset \Omega$ 
\begin{equation*}
||u^\varepsilon - \tilde{v}^\varepsilon||^2_{L^2(\Omega')} = 
\sum_{k=1}^{N(\varepsilon)}\int_{E_k^{\varepsilon}}|\tilde{v}^\varepsilon(x) - \tilde{v}_k^\varepsilon|^2 dx.
\end{equation*}
By Poincar\'e inequality, for every $k$ we have 
\begin{equation*}
\int_{E_k^{\varepsilon}}|\tilde{v}^\varepsilon(x) - \tilde{v}_k^\varepsilon|^2 dx 
\leq c \,(\mathcal{L}^n (E_k^{\varepsilon}))^{2/n} \int_{E_k^{\varepsilon}}|D \tilde{v}^\varepsilon(x)|^2 dx,
\end{equation*}
and $\mathcal{L}^n (E_k^{\varepsilon})$ is of order $\varepsilon^n$, hence 
\begin{equation*}
||\,u^\varepsilon - \tilde{v}^\varepsilon||^2_{L^2(\Omega')}  \leq c \,\varepsilon^2 
\sum_{k=1}^{N(\varepsilon)}\int_{E_k^{\varepsilon}}|D \tilde{v}^\varepsilon(x)|^2 dx 
\leq c \,\varepsilon^2 \int_{\Omega'}|D \tilde{v}^\varepsilon(x)|^2 dx \leq c \,\varepsilon^2. 
\end{equation*}
This entails that $u^\varepsilon \rightarrow u$ strongly in $L^2(\Omega')$ and hence 
strongly in $L^2_{loc}(\Omega)$. Therefore, $(i)$ is proved. 

Now, we prove $(ii)$.
Let us write explicitly
the expression of $\mathcal{F}^{\varepsilon}(u^{\varepsilon})$, 
\begin{align*}
\mathcal{F}^{\varepsilon}(u^{\varepsilon}) =&\,\int_{\Omega} |\nabla u^{\varepsilon}|^2 dx + \int_{S_{u^{\varepsilon}}}f_{\alpha_\varepsilon}
\Big(\frac{x}{\varepsilon}\Big) \,d \mathcal{H}^{n - 1}(x) = 
\int_{\Omega\setminus\tilde{E}^{\varepsilon}} |\nabla u^{\varepsilon}|^2 dx + 
\alpha_\varepsilon \mathcal{H}^{n - 1}(S_{u^{\varepsilon}}) \\
= & \, \int_{\Omega \setminus \tilde{E}^{\varepsilon}} |D \tilde{v}^{\varepsilon}|^2 dx + 
\alpha_\varepsilon \mathcal{H}^{n - 1}(S_{u^{\varepsilon}})
= \mathcal{G}^{\varepsilon}(\tilde{v}^{\varepsilon}) + 
\alpha_\varepsilon \mathcal{H}^{n - 1}(S_{u^{\varepsilon}}\cap \tilde{E}^{\varepsilon}).
\end{align*}
Notice that if we show that $\alpha_\varepsilon 
\mathcal{H}^{n - 1}(S_{u^{\varepsilon}}\cap \tilde{E}^{\varepsilon}) = o(\varepsilon)$ 
as $\varepsilon \rightarrow 0$, then $(ii)$ follows directly. Actually, we have
\begin{equation*}
\alpha_\varepsilon \mathcal{H}^{n - 1}(S_{u^{\varepsilon}}\cap \tilde{E}^{\varepsilon}) \leq 
\alpha_\varepsilon N(\varepsilon) \,P(E^\varepsilon,Q^\varepsilon) = 
C\,\alpha_\varepsilon \frac{1}{\varepsilon^n}\,\varepsilon^{n - 1}  
=  C\,\frac{\alpha_\varepsilon}{\varepsilon},
\end{equation*}
and $\frac{\alpha_\varepsilon}{\varepsilon}= o(\varepsilon)$ as $\varepsilon \rightarrow 0$ by assumption.
\end{proof}


\section{Supercritical regime: stiffer inclusions}

In this Section we consider the case $\alpha_\varepsilon >> \varepsilon$. We have 
previously shown that for $\alpha_\varepsilon << \varepsilon$ configurations exhibiting a high 
number of discontinuities are favoured by the energy. We will prove that on the contrary in this regime the 
energy penalizes the presence of jumps in the displacements.

Before studying this case, we state and prove some technical lemmas
which will be used in the following.

\begin{lem}\label{tech}
Let us consider a sequence of measurable functions $a_k:\Omega \rightarrow \mathbb{R}_+$ such that
\begin{equation*}
a_k \rightarrow a \quad \hbox{in measure}.
\end{equation*}
Then, for every $u\in L^2(\Omega)$ and for every sequence $(u_{k})\subset L^2(\Omega)$ such that 
$$u_k \rightharpoonup u \quad \hbox{weakly in } L^{2}(\Omega),$$ 
it turns out that 
\begin{equation*}
\int_{\Omega}a u^2 dx \leq \liminf_{k \rightarrow + \infty} \int_{\Omega}a_k u_k^2 dx.
\end{equation*}
\end{lem}

\begin{proof}
Let $u\in L^2(\Omega)$ and $u_k \rightharpoonup u$ weakly in $L^{2}(\Omega)$.

We can extract a subsequence $(k_j)$ such that
\begin{equation}\label{liminflim}
\liminf_{k \rightarrow + \infty} \int_{\Omega}a_k u_k^2 dx = 
\lim_{j \rightarrow + \infty} \int_{\Omega}a_{k_j} u^2_{k_j} dx.
\end{equation}
From the convergence in measure of $a_k$ to $a$ we deduce that for every $\eta > 0$ 
there exists a measurable set $D_{\eta}\subset \Omega$ such that 
$\mathcal{L}^n(D_{\eta}) < \eta$ and 
\begin{equation*}
|a_{k_{j_{i}}} - a| \leq \frac{1}{i} \quad \hbox{a.e. on } \Omega \setminus D_{\eta}
\end{equation*}
for a suitable subsequence $(a_{k_{j_{i}}})$ of $(a_{k_j})$. By (\ref{liminflim}) we get
\begin{align*}
\liminf_{k \rightarrow + \infty} \int_{\Omega}a_k u_k^2 dx = &\,
\lim_{i \rightarrow + \infty}\int_{\Omega}a_{k_{j_i}} u_{k_{j_i}}^2 dx \geq 
\lim_{i \rightarrow + \infty}\int_{\Omega\setminus D_{\eta}}a_{k_{j_i}} u_{k_{j_i}}^2 dx\\ 
&\,\geq \liminf_{i \rightarrow + \infty} \bigg\{\int_{\Omega\setminus D_{\eta}}a\,u_{k_{j_i}}^2 dx - 
\frac{1}{i}\,\int_{\Omega} u_{k_{j_i}}^2 dx \bigg\}.
\end{align*}
Using the lower semicontinuity of the functional $L^2(\Omega)\ni u \rightarrow \int_{\Omega\setminus D_\eta} a\,u^2 dx$ 
with respect to the weak topology of $L^2$, we have
\begin{equation*}
\liminf_{k \rightarrow + \infty} \int_{\Omega}a_k u_k^2 dx \geq \int_{\Omega\setminus D_{\eta}}a u^2 dx
\end{equation*}
for every $\eta > 0$. Letting $\eta \rightarrow 0$ the claim follows.
\end{proof}

In the next lemma we state and prove a $\Gamma$-convergence result for an auxiliary functional 
that will appear in the proof of the main theorem of this section.

\begin{lem}\label{GconGh}
Let us fix \, $0< \bar{\delta} <\delta < \frac{1}{2}$ such that $Q_{\delta} \subset\subset Q_{\bar{\delta}}$.
For every $h\in \mathbb{N}$, let \,$\mathcal{I}^h: L^2(Q_{\bar{\delta}})
\rightarrow [0,+\infty]$ be the functional defined as
\begin{equation*}
\mathcal{I}^h (w):= 
\begin{cases}
\displaystyle \int_{Q_{\bar{\delta}}} |\nabla w|^2 dx + \mathcal{H}^{n-1}(S_w) \quad &\hbox{if } w\in 
SBV^2(Q_{\bar{\delta}}), S_{w}\subset Q_\delta, \mathcal{H}^{n-1}(S_w) \leq \frac{1}{h},\\
+\infty & \hbox{otherwise in } L^2(Q_{\bar{\delta}}).
\end{cases}
\end{equation*}
Then the sequence $\mathcal{I}^h$ $\Gamma$-converges with respect to the strong 
topology of $L^2$ to the functional $\mathcal{I}: L^2(Q_{\bar{\delta}})
\rightarrow [0,+\infty]$ given by
\begin{equation*}
\mathcal{I}(w):= 
\begin{cases}
\displaystyle \int_{Q_{\bar{\delta}}} |D w|^2 dx \quad &\hbox{if } w\in 
H^1(Q_{\bar{\delta}}),\\
+\infty & \hbox{otherwise in } L^2(Q_{\bar{\delta}}).
\end{cases}
\end{equation*}
\end{lem}

\begin{proof}
Let $w\in L^2(Q_{\bar{\delta}})$ and let $(w_h)$ be a sequence converging to 
$w$ strongly in $L^2$ and having equibounded energy $\mathcal{I}^h$. 
We claim that $w\in H^{1}(Q_{\bar{\delta}})$ and that  
\begin{equation}\label{bfbg}
\liminf_{h\rightarrow +\infty}\,\mathcal{I}^h(w_h) \geq \mathcal{I}(w).
\end{equation}
Without loss of generality we can assume that $||w_h||_{L^{\infty}} \leq c< + \infty$.
Indeed, if the claim (\ref{bfbg}) is proved in this case, then we can recover the general result 
in the following way. 
Let $w\in L^2(Q_{\bar{\delta}})$ and $(w_h)\subset L^2(Q_{\bar{\delta}})$ 
converging to $w$ strongly in $L^2$ and having equibounded energy. For every 
$l\in \mathbb{N}$ let us define $T_{l}(w_h):= \big(w_h \wedge l \big)\vee (- l)$.
Since $T_{l}(w_h)$ converges to $T_{l}w$ strongly in $L^2$ as $h \rightarrow +\infty$ 
and $||T_{l}(w_h)||_{L^{\infty}} \leq l$, we have by (\ref{bfbg}) that 
$T_{l}w \in H^1(Q_{\bar{\delta}})$ and
\begin{equation*}
\liminf_{h\rightarrow +\infty}\mathcal{I}^h \big(T_{l}(w_h)\big)
\geq \mathcal{I}(T_{l}w).
\end{equation*}
Now, by
\begin{equation*}
\mathcal{I}^h\big(T_{l}(w_h)\big) \leq \mathcal{I}^h(w_h),
\end{equation*}
we have that for every $l\in \mathbb{N}$
\begin{equation}\label{gsbv}
\liminf_{h \rightarrow +\infty}\mathcal{I}^h(w_h) \geq \mathcal{I}(T_{l}w).
\end{equation}
Since $(w_h)$ has equibounded energy, this inequality implies that 
$(T_{l}w)$ is equibounded in $H^1(Q_{\bar{\delta}})$. 
Hence, there exists a subsequence $(l_k)$ and a function $v\in H^1(Q_{\bar{\delta}})$
such that $T_{l_k}w$ converges to $v$ weakly in $H^1(Q_{\bar{\delta}})$, hence strongly 
in $L^2(Q_{\bar{\delta}})$, as $k\rightarrow + \infty$. 
From the uniqueness of the limit, since $w$ is the pointwise limit of $T_{l}w$, it follows that $v = w$, 
which entails that $w\in H^1(Q_{\bar{\delta}})$.

In view of these remarks and of the lower semicontinuity of the Dirichlet functional, 
in (\ref{gsbv}) we obtain the chain of inequalities
\begin{align*}
\liminf_{h\rightarrow +\infty}\mathcal{I}^h(w_h)
\geq \limsup_{l\rightarrow + \infty} \mathcal{I}(T_{l}w) 
\geq \limsup_{k\rightarrow + \infty} \mathcal{I}(T_{l_k}w)
\geq \liminf_{k\rightarrow + \infty} \mathcal{I}(T_{l_k}w) \geq \mathcal{I}(w),
\end{align*}
which is exactly (\ref{bfbg}). 

So, from now on we will assume that $||w_h||_{L^{\infty}} \leq c< + \infty$.
Under this further assumption we can apply directly Ambrosio's compactness and 
lower semicontinuity theorems (see for instance \cite{Amb} and \cite{AM}) in 
order to deduce the compactness for the sequence 
$(w_h)$ having equibounded energy and the liminf inequality. The fact that 
$\mathcal{H}^{n-1}(S_{w_h}) \leq \frac{1}{h}$ ensures in particular that the 
limit function belongs to the Sobolev space $H^1$.

Finally, the existence of a recovery sequence for a function $w\in H^1(Q_{\bar{\delta}})$ 
follows immediately by taking $w_h = w$ for every $h\in \mathbb{N}$.
\end{proof}


Next lemma contains a $\Gamma$-convergence result for the same functionals as in Lemma \ref{GconGh}, 
but taking into account Dirichlet boundary conditions.

\begin{lem}\label{bcs}
Let $(\varphi_h), \varphi \in H^{1/2}(\partial Q_{\bar{\delta}})$ be such that 
$\varphi_h \rightarrow \varphi$ strongly in $H^{1/2}(\partial Q_{\bar{\delta}})$.
For every $h\in \mathbb{N}$, let $\mathcal{I}_{\varphi_h}^h: L^2(Q_{\bar{\delta}})
\rightarrow [0,+\infty]$ be the functional defined by
\begin{equation}\label{defnbc}
\mathcal{I}_{\varphi_h}^h (w):= \left\{
\begin{array}{lll}
\vspace{-.2cm}
\displaystyle \int_{Q_{\bar{\delta}}} |\nabla w|^2 dx + \mathcal{H}^{n-1}(S_w)\quad &\hbox{if } w\in 
SBV^2(Q_{\bar{\delta}}), S_{w}\subset Q_\delta,\mathcal{H}^{n-1}(S_w) \leq \frac{1}{h},\\
\vspace{.15cm}
& \,\,\,\, w=\varphi_h \, \mbox{on } \partial Q_{\bar{\delta}},\\
+\infty & \hbox{otherwise in } L^2(Q_{\bar{\delta}}).
\end{array}
\right.
\end{equation}
Then the sequence $(\mathcal{I}_{\varphi_h}^h)$ $\Gamma$-converges with respect to the strong 
topology of $L^2$ to the functional $\mathcal{I}_{\varphi}: L^2(Q_{\bar{\delta}})
\rightarrow [0,+\infty]$ given by
\begin{equation*}
\mathcal{I}_{\varphi}(w):= 
\begin{cases}
\displaystyle \int_{Q_{\bar{\delta}}} |D w|^2 dx \quad &\hbox{if } w\in 
H^1(Q_{\bar{\delta}}),\, w = \varphi \, \mbox{on } \partial Q_{\bar{\delta}},\\
+\infty & \hbox{otherwise in } L^2(Q_{\bar{\delta}}).
\end{cases}
\end{equation*} 
\end{lem}

\begin{proof}
\textit{First step: proof of compactness and liminf.}
Let $(w_h), w \in L^2(Q_{\bar{\delta}})$ be such that $w_h \rightarrow w$ strongly in $L^2$ and 
$\mathcal{I}_{\varphi_h}^h (w_h) \leq c < + \infty$.
From the equality $\mathcal{I}_{\varphi_h}^h (w_h) = \mathcal{I}^h (w_h)$ and the previous 
lemma, we get that $w \in H^1(Q_{\bar{\delta}})$; moreover,
\begin{equation*}
 \liminf_{h \rightarrow \infty} \mathcal{I}_{\varphi_h}^h (w_h) = 
\liminf_{h \rightarrow \infty}\mathcal{I}^h (w_h) \geq \mathcal{I}(w). 
\end{equation*}
It remains to show that $w = \varphi$ on $\partial Q_{\bar{\delta}}$. First of all
we can notice that the bound $\mathcal{I}_{\varphi_h}^h (w_h) \leq c < + \infty$ 
implies that $w_h = \varphi_h$ on $\partial Q_{\bar{\delta}}$. Moreover we have
$||w_h||_{H^1(Q_{\bar{\delta}} \setminus Q_\delta)} \leq c$, hence 
$w_h \rightharpoonup w$ weakly in $H^1(Q_{\bar{\delta}} \setminus Q_\delta)$. 
This convergence entails in particular the convergence of the traces on $\partial Q_{\bar{\delta}}$,
that is,
\begin{equation}\label{tracce}
 \varphi_h = (w_h)_{|\partial Q_{\bar{\delta}}} \rightarrow w_{|\partial Q_{\bar{\delta}}} 
\quad \hbox{strongly in } L^2(\partial Q_{\bar{\delta}}).
\end{equation}
Since $\varphi_h \rightarrow \varphi$ strongly in $H^{1/2}(\partial Q_{\bar{\delta}})$, from (\ref{tracce}) 
we get the equality $w = \varphi$ on $\partial Q_{\bar{\delta}}$. 

\textit{Second step: limsup.} Let $w \in H^1(Q_{\bar{\delta}})$ be such that $w = \varphi$ on $\partial Q_{\bar{\delta}}$. 
The surjectivity of the trace operator onto $H^{1/2}$ and the continuity of the inverse ensure that for every $h\in \mathbb{N}$ there exists 
$v_h \in H^1(Q_{\bar{\delta}})$ verifying the equality $v_h = \varphi_h - \varphi$ 
on $\partial Q_{\bar{\delta}}$ and the bound
\begin{equation*}
 ||v_h||_{H^1(Q_{\bar{\delta}})} \leq c\, ||\varphi_h - \varphi||_{H^{1/2}(\partial Q_{\bar{\delta}})}.
\end{equation*}
From the assumption we have $v_h \rightarrow 0$ strongly in $H^1$. Let us define the sequence $w_h = w + v_h$. 
It turns out that $w_h = \varphi_h$ on $\partial Q_{\bar{\delta}}$ and that $w_h \rightarrow w$ strongly in $H^1$.
Therefore $w_h$ is a recovery sequence for $\mathcal{I}^h_{\varphi_h}$.
\end{proof}


\noindent
Now we are ready to state and prove the main result of this Section.

Define the functional $\mathcal{F}^{\infty}: L^{2}(\Omega) \rightarrow [0, + \infty]$ as
\begin{equation*}
\mathcal{F}^{\infty}(u) = 
\begin{cases}
\displaystyle \int_{\Omega}|D u|^2\,dx \quad &\hbox{in } H^1(\Omega),\\
+ \infty \quad  &\hbox{otherwise in } L^2(\Omega).
\end{cases}
\end{equation*}
We will show that $\mathcal{F}^{\infty}$ is the $\Gamma$-limit of the sequence 
$(\mathcal{F}^{\varepsilon})$ in this case, that is, when $\alpha_\varepsilon >> \varepsilon$.


\begin{thm}[Bound from below]\label{bfbw}
Let $u\in L^2(\Omega)$ and let $(u^\varepsilon)$ be a sequence converging to $u$ strongly
in $L^2$ and having equibounded energy $\mathcal{F}^{\varepsilon}$. Then $u\in H^{1}(\Omega)$ and 
\begin{equation}\label{bfb}
\liminf_{\varepsilon\rightarrow 0}\mathcal{F}^{\varepsilon}(u^{\varepsilon})
\geq \mathcal{F}^{\infty}(u).
\end{equation}
\end{thm}

\begin{proof}
We remark that, as $\mathcal{F}^{\varepsilon}(u^\varepsilon)$ is bounded, the 
functions $u^\varepsilon$ can have jumps only in the set $\tilde{E}^{\varepsilon}\cup 
\tilde{F}^{\varepsilon}$ defined in (\ref{subd}). 

We now classify the cubes $Q^{\varepsilon}_k$ according to the measure 
of the jump set that they contain.
More precisely, let us introduce a positive parameter $\beta > 0$ that will be chosen later in a suitable way.
We say that a cube $Q^{\varepsilon}_k$ is \textit{good} whenever 
$\mathcal{H}^{n-1}\big(S_{u^{\varepsilon}}\cap Q^{\varepsilon}_k \big) 
\leq \beta\,\varepsilon^{n-1}$, and 
\textit{bad} otherwise and we denote with $N_1(\varepsilon)$ and $N_2(\varepsilon)$ the number of 
\textit{good} and \textit{bad} cubes, respectively. 
First of all we can notice that, by the fact that the sequence $(u^\varepsilon)$ has equibounded energy, 
we have in particular that there exists a constant $c>0$ such that 
$\alpha_\varepsilon \mathcal{H}^{n-1}(S_{u^\varepsilon}) \leq c$. 
From this we deduce an important bound for the number of bad cubes, that is 
$N_2(\varepsilon) \leq \dfrac{c}{\alpha_\varepsilon \varepsilon^{n-1}}$. 
We can write, from (\ref{domom}),

\begin{equation}\label{domomg}
\Omega = \Bigg(\bigcup_{k=1}^{N_1(\varepsilon)} Q^{\varepsilon}_k\Bigg)\cup 
\Bigg(\bigcup_{k=1}^{N_2(\varepsilon)} Q^{\varepsilon}_k\Bigg)\cup R(\varepsilon)
=: Q_g ^{\varepsilon}\cup Q_b^{\varepsilon}\cup R(\varepsilon).
\end{equation}

\textit{First step: energy estimate on good cubes.} Let $Q^{\varepsilon}_k$ be a good cube and consider
\begin{equation}\label{fc}
\mathcal{F}^{\varepsilon}\big(u^\varepsilon,Q^{\varepsilon}_k\big) := \int_{Q^{\varepsilon}_k} |\nabla u^\varepsilon|^2 dx + \alpha_\varepsilon
\mathcal{H}^{n-1}\big(S_{u^{\varepsilon}}\cap Q^{\varepsilon}_k \big).
\end{equation}
Define the function $v^\varepsilon$ in the unit cube $Q_k$ as
$u^\varepsilon(\varepsilon\,y)=:\sqrt{\alpha_\varepsilon \varepsilon}\,v^\varepsilon(y)$.
In terms of $v^\varepsilon$, (\ref{fc}) becomes
\begin{equation}\label{step}
\mathcal{F}^{\varepsilon}\big(u^\varepsilon,Q^{\varepsilon}_k\big) = \alpha_\varepsilon \varepsilon^{n-1} \bigg\{\int_{Q_k} |\nabla v^\varepsilon|^2 dx + 
\mathcal{H}^{n-1}(S_{v^{\varepsilon}}\cap Q_k)\bigg\},
\end{equation}
with $\mathcal{H}^{n-1} (S_{v^{\varepsilon}}\cap Q_k) \leq \beta$.
In other words, by means of a change of variables we have reduced the problem to the study of 
the Mumford-Shah functional over a fixed domain, with some constraints on the jump set. 
From now on we will omit the subscript $k$.
Let $\bar{\delta},\hat{\delta}$ be such that $Q_{\delta}\subset\subset Q_{\bar{\delta}}\subset\subset Q_{\hat{\delta}}\subset\subset Q$. 

Let us consider the problem of finding local minimizers for the Mumford-Shah functional under the 
required conditions, that is
\begin{equation*}
\begin{array}{ll}
&\hbox{(LMS)} \,\, \displaystyle \hbox { loc}\min\bigg\{\int_{Q_{\hat{\delta}}} |\nabla w|^2 dx + 
\mathcal{H}^{n-1}(S_{w}) : w \in SBV^2(Q_{\hat{\delta}}), S_{w}\subset E\cup F,\,\mathcal{H}^{n-1}(S_{w}) \leq \beta\bigg\}.
\end{array}
\end{equation*}
According to the definition given in \cite{DMM}, we recall that a local minimizer is a function which minimizes 
the given functional with respect to all perturbations with compact support.
Let us denote by $\mathcal{M}_{\beta}$ the class of solutions of (LMS).

For a given $\hat{v}\in \mathcal{M}_{\beta}$, let us consider the function $\tilde{v}$ 
solving
\begin{equation*}
\hbox{(Dir)} \quad 
\begin{cases}
\Delta w = 0  & \hbox{in } Q_{\bar{\delta}}\\
w = \hat{v} & \hbox{in } Q_{\hat{\delta}} \setminus Q_{\bar{\delta}}.
\end{cases}
\end{equation*}

We want to prove that for every $\eta > 0$ there exists $\beta > 0$ such that for every 
$\hat{v}\in \mathcal{M}_{\beta}$ and for the corresponding $\tilde{v}$ we have
\begin{equation}\label{veroclaim}
\int_{Q_{\hat{\delta}}} |\nabla \tilde{v}|^2 dx \leq 
(1 + \eta)\int_{Q_{\hat{\delta}}}|\,\nabla \hat{v}|^2 dx. 
\end{equation}
Hence we will take such a $\beta$ in the definition of good and bad cubes.

Let us prove (\ref{veroclaim}) by contradiction. Suppose (\ref{veroclaim}) is false. Then there exists
$\eta>0$ such that for every $\beta>0$ there exists $\hat{v} \in \mathcal{M}_{\beta}$ and 
a corresponding $\tilde{v}$ for which
\begin{equation}\label{negoclaim}
\int_{Q_{\hat{\delta}}} |\nabla \tilde{v}|^2 dx >  
(1 + \eta)\int_{Q_{\hat{\delta}}}|\nabla \hat{v}|^2 dx. 
\end{equation}

\noindent
In particular (\ref{negoclaim}) implies that for every $h > 0$ there exists $\hat{v}_h 
\in \mathcal{M}_{\frac{1}{h}}$ and $\tilde{v}_h$ solution of (Dir) with $\hat{v}$ replaced by 
$\hat{v}_h$ for which
\begin{equation}\label{negoclaimh}
\int_{Q_{\hat{\delta}}} |\nabla \tilde{v}_h|^2 dx >  
(1 + \eta)\int_{Q_{\hat{\delta}}}|\nabla \hat{v}_h|^2 dx.
\end{equation}
Since $Q_{\hat{\delta}}= \big(Q_{\hat{\delta}}\setminus Q_{\bar{\delta}}\big) \cup Q_{\bar{\delta}}$, 
we can split the previous integrals and, using the fact that $\tilde{v}_h = \hat{v}_h$ in 
$Q_{\hat{\delta}}\setminus Q_{\bar{\delta}}$ we obtain from (\ref{negoclaimh})
\begin{equation}\label{mah}
\int_{Q_{\bar{\delta}}} |\nabla \tilde{v}_h|^2 dx >  (1 + \eta)\int_{Q_{\bar{\delta}}}|\nabla \hat{v}_h|^2 dx + 
\eta \int_{Q_{\hat{\delta}}\setminus Q_{\bar{\delta}}}|\nabla \hat{v}_h|^2 dx.
\end{equation}
Since the problem defining $\tilde{v}_h$ is linear, we can normalize the left-hand side of (\ref{mah}), 
so that we can assume 
\begin{equation}\label{normalize}
1 = \int_{Q_{\bar{\delta}}} |\nabla \tilde{v}_h|^2 dx >  (1 + \eta)\int_{Q_{\bar{\delta}}}|\nabla \hat{v}_h|^2 dx + 
\eta \int_{Q_{\hat{\delta}}\setminus Q_{\bar{\delta}}}|\nabla \hat{v}_h|^2 dx.
\end{equation}
This means in particular that
\begin{equation}\label{bouen}
 \int_{Q_{\hat{\delta}}}|\nabla \hat{v}_h|^2 dx \leq  \frac{1}{\eta} < + \infty. 
\end{equation}
Without loss of generality we can assume that $\int_{Q_{\hat{\delta}}\setminus Q_{\delta}}\hat{v}_h dx =0$; 
therefore, since $S_{\hat{v}_h} \subset Q_{\delta}$, (\ref{bouen}) implies that
$||\hat{v}_h||_{H^1(Q_{\hat{\delta}}\setminus Q_{\delta})} \leq c$.
Using the fact that $\hat{v}_h$ is harmonic in $Q_{\hat{\delta}}\setminus Q_{\delta}$ we get the convergence of
the traces of $\hat{v}_h$ on $\partial Q_{\bar{\delta}}$, that is
\begin{equation}\label{convtra}
 \varphi_{h}:= (\hat{v}_h)_{|\partial Q_{\bar{\delta}}} \rightarrow \varphi 
\quad \hbox{strongly in } H^{1/2}(\partial Q_{\bar{\delta}}).
\end{equation}
At this point, let us consider the following problems:
\begin{equation*}
(\hbox{Dir})_{\varphi_h} \quad 
\begin{cases}
\Delta w = 0  & \hbox{in } Q_{\bar{\delta}}\\
\vspace{.1cm}
w = \varphi_h & \hbox{on } \partial Q_{\bar{\delta}},
\end{cases}
\quad \quad
(\hbox{Dir})_{\varphi} \quad 
\begin{cases}
\Delta w = 0  & \hbox{in } Q_{\bar{\delta}}\\
\vspace{.1cm}
w = \varphi & \hbox{on } \partial Q_{\bar{\delta}}.
\end{cases}
\end{equation*}
Clearly, $\tilde{v}_h$ is the only solution to $\hbox{(Dir)}_{\varphi_h}$ for every $h$.
Let us call $\tilde{v}$ the solution to $\hbox{(Dir)}_{\varphi}$. From (\ref{convtra}) it 
turns out that $\tilde{v}_h \rightarrow \tilde{v}$ strongly in $H^1(Q_{\bar{\delta}})$, hence,
\begin{equation}\label{gradarm}
1 =  \int_{Q_{\bar{\delta}}} |\nabla \tilde{v}_h|^2\,dx \rightarrow \int_{Q_{\bar{\delta}}} |\nabla \tilde{v}|^2\,dx = 1.
\end{equation}
Notice that the functions $\hat{v}_h$ defined by the minimum problem (LMS) are absolute 
minimizers of the same functional over the same class once we fix the boundary data $\varphi_h$. Therefore they are 
absolute minimizers for the functional $\mathcal{I}^h_{\varphi_h}$ defined in (\ref{defnbc}).
The $\Gamma$-convergence result proved in Lemma \ref{bcs} gives the $L^2$ 
convergence of the sequence $\hat{v}_h$ to the only minimizer of the functional $\mathcal{I}_\varphi$, 
that is exactly $\tilde{v}$, and the convergence of the energies.

Now, if we let $h \rightarrow +\infty$ in (\ref{normalize}) we obtain that
\begin{equation*}
 1 = \int_{Q_{\bar{\delta}}} |\nabla \tilde{v}|^2 dx \geq  (1 + \eta)\int_{Q_{\bar{\delta}}}|\nabla \tilde{v}|^2 dx,
\end{equation*}
which gives the contradiction, therefore (\ref{veroclaim}) is proved. 

Let $\eta>0$ be fixed; we choose $\beta > 0$ such that the property (\ref{veroclaim}) is satisfied and 
for every $\varepsilon>0$ we consider the problem 
\begin{equation*}
\begin{array}{ll}
\vspace{-.2cm}
\hbox{(MS)} \,\, \displaystyle \min\bigg\{\int_{Q_k^{\hat{\delta}}} |\nabla w|^2 dx + 
\mathcal{H}^{n-1}(S_{w}) :& w \in SBV^2(Q_k^{\hat{\delta}}), S_{w}\subset E\cup F, \\
&\,\mathcal{H}^{n-1}(S_{w}) \leq \beta, w = v^\varepsilon \, \hbox{on } \partial Q_k^{\hat{\delta}}\bigg\}.
\end{array}
\end{equation*}
For a minimizer $\hat{v}^\varepsilon$ of (MS), let $\tilde{v}^\varepsilon$ be the corresponding function
defined by (Dir), with $\hat{v}$ replaced by $\hat{v}^\varepsilon$. We have that, as before, 
\begin{equation}\label{veroclaimeps}
\int_{Q_k^{\hat{\delta}}} |\nabla \tilde{v}^\varepsilon|^2 dx \leq 
(1 + \eta)\int_{Q_k^{\hat{\delta}}}|\nabla \hat{v}^\varepsilon|^2 dx. 
\end{equation}
Hence, in particular, 
\begin{align}\label{stepp}
\int_{Q_{\hat{\delta},k}} |\nabla v^\varepsilon|^2 dx + 
\mathcal{H}^{n-1}(S_{v^{\varepsilon}}\cap Q_{\hat{\delta},k})
&\geq\int_{Q_{\hat{\delta},k}} |\nabla \hat{v}^\varepsilon|^2 dx + 
\mathcal{H}^{n-1}(S_{\hat{v}^{\varepsilon}}\cap Q_{\hat{\delta},k}) \nonumber\\
&\geq \bigg(1 - \frac{\eta}{1 + \eta}\bigg) \int_{Q_{\hat{\delta},k}} |\nabla \tilde{v}^\varepsilon|^2 dx,
\end{align}
where $v^\varepsilon$ is the function in (\ref{step}). 
Now define $\tilde{u}^\varepsilon$ as $\tilde{u}^\varepsilon(\varepsilon\,y):=\sqrt{\alpha_\varepsilon \varepsilon}\,\,\tilde{v}^\varepsilon(y)$. By (\ref{step}) and (\ref{stepp}) we obtain 
\begin{equation}\label{fc2}
\int_{Q_{\hat{\delta},k}^{\varepsilon}} |\nabla u^\varepsilon|^2 dx + \alpha_\varepsilon
\mathcal{H}^{n-1}\big(S_{u^{\varepsilon}}\cap\big(Q_{\hat{\delta},k}^{\varepsilon}\big)\big) \geq 
\bigg(1 - \frac{\eta}{1 + \eta}\bigg)\int_{Q_{\hat{\delta},k}^{\varepsilon}}|\nabla \,\tilde{u}^\varepsilon|^2 dx.
\end{equation}

\textit{Second step: energy estimate on bad cubes.} Let $Q^{\varepsilon}_k$ be a bad cube. This means that 
$\mathcal{H}^{n-1}\big(S_{u^{\varepsilon}}\cap Q^{\varepsilon}_k\big) > 
\beta\,\varepsilon^{n-1}$.
First of all, recall that we have a control on the number of bad cubes, that is, 
$N_2(\varepsilon) \leq \dfrac{c}{\alpha_\varepsilon \varepsilon^{n-1}}$.
The idea is to use the obvious inequality
\begin{equation*}
\int_{Q^{\varepsilon}_k} |\nabla  u^\varepsilon|^2 dx + \alpha_\varepsilon
\mathcal{H}^{n-1}\big(S_{u^{\varepsilon}}\cap Q^{\varepsilon}_k\big) \geq 
\int_{Q^{\varepsilon}_k} \chi_{\delta}^{\varepsilon}\,|\nabla \check{u}^\varepsilon|^2 dx,
\end{equation*}
where $\chi_{\delta}^{\varepsilon}$ is the characteristic function of the set $Q^{\varepsilon}_k\setminus 
Q_{\delta,k}^{\varepsilon}$ and the function $\check{u}^\varepsilon$ coincides with $u^\varepsilon$ in $Q^{\varepsilon}_k\setminus Q_{\delta,k}^{\varepsilon}$ and is extended to $Q_{\delta,k}^{\varepsilon}$ in a 
way that keeps its $H^1$ norm bounded.

\textit{Third step: final estimate. }
Let us define a new sequence $w^\varepsilon \in SBV^2(\Omega)$ as
\begin{equation*}
w^\varepsilon := \left\{
\begin{array}{ll}
\vspace{.2cm}
\tilde{u}^\varepsilon    \quad &\hbox{in } Q_{g}^{\hat{\delta},\varepsilon},\\
\vspace{.2cm}
u^\varepsilon    \quad &\hbox{in } \big(Q_{g}^{\varepsilon} \setminus Q_{\hat{\delta},g}^{\varepsilon}\big)\cup R(\varepsilon),\\
\check{u}^\varepsilon \quad &\hbox{in } Q_{b}^{\varepsilon},
\end{array}
\right.
\end{equation*}
where $Q_{g}^{\varepsilon}, Q_{b}^{\varepsilon}$ and $R(\varepsilon)$ are 
given in (\ref{domomg}) and $Q_{\hat{\delta},g}^{\varepsilon}$ denotes the set
\begin{equation*}
 Q_{\hat{\delta},g}^{\varepsilon}:= \bigcup_{k=1}^{N_1(\varepsilon)} Q_{\hat{\delta},k}^{\varepsilon}.
\end{equation*}

Define also the function $a^\varepsilon: \Omega \rightarrow \mathbb{R}$ as
\begin{equation*}
a^\varepsilon(x) := 
\begin{cases}
0 \quad & \hbox{in } Q_{\delta,b}^{\varepsilon},\\
1 \quad & \hbox{otherwise in } \Omega.
\end{cases}
\end{equation*}
From what we proved in the previous steps we can write
\begin{equation}\label{uff}
\mathcal{F}^{\varepsilon}(u^{\varepsilon}) \geq \bigg(1 - \frac{\eta}{1 + \eta}\bigg)
\int_{\Omega} a^\varepsilon(x)\,|\nabla w^\varepsilon|^2 dx.
\end{equation}
It remains to apply Lemma \ref{tech} to (\ref{uff}). 
First of all we show the convergence of $a^\varepsilon$. We have
\begin{equation*}
\int_{\Omega}|\,a^\varepsilon - 1|\,dx = \mathcal{L}^n 
\big(Q_{\delta,b}^{\varepsilon}\big) = N_2(\varepsilon)\,\varepsilon^n \mathcal{L}^n(Q_\delta) \leq c\, \frac{\varepsilon}{\alpha_\varepsilon},
\end{equation*}
hence $a^\varepsilon \rightarrow 1$ strongly in $L^1(\Omega)$. Once we prove that 
$w^\varepsilon \rightharpoonup u$ weakly in $H^1(\Omega)$, it turns out that
\begin{equation*}
\liminf_{\varepsilon\rightarrow 0}\mathcal{F}^{\varepsilon}(u^{\varepsilon}) \geq \bigg(1 - \frac{\eta}{1 + \eta}\bigg)
\int_{\Omega} |D u|^2 dx,
\end{equation*}
and the thesis follows letting $\eta$ converge to zero.

\textit{Fourth step: convergence of $w^\varepsilon$.} 
First of all it is clear from (\ref{uff}) and the choice of $\check{u}^{\varepsilon}$ that $||\nabla w^\varepsilon||_{L^2(\Omega)} \leq c$. 
Then, as in the proof of Theorem \ref{GammaG}, the fact that $w^\varepsilon$ and
$u^\varepsilon$ coincide in a set with positive measure ensures the convergence.

\end{proof}


\begin{thm}[Bound from above]
For every $u\in H^1(\Omega)$ there exists a sequence $(u^\varepsilon)$ such that 
\begin{align*}
&(i) \quad u^\varepsilon \rightarrow u \quad \hbox{strongly in } L^2(\Omega), \\
&(ii) \quad \lim_{\varepsilon\rightarrow 0}\mathcal{F}^{\varepsilon}(u^{\varepsilon}) = \mathcal{F}^{\infty}(u).
\end{align*}
\end{thm}

\begin{proof}
The thesis follows trivially by choosing $u^\varepsilon = u$ for every $\varepsilon > 0$. 
\end{proof}


\section{Critical regime: intermediate case}

In this section we will analyze the case in which the fragility coefficient of the inclusions in the material 
and the size $\varepsilon$ of the periodic structure are of the same order. We can assume, without loss of generality, that $\alpha_\varepsilon = \varepsilon$.
So, the functional we are interested in is given by
\begin{equation*}
\mathcal{F}^\varepsilon(u) =  
\begin{cases}
\displaystyle\int_{\Omega}|\nabla u|^2 dx + \varepsilon\,\mathcal{H}^{n - 1}(S_{u}) \quad &\hbox{if } u \in SBV^2(\Omega), S_{u}\subset \tilde{E}^{\varepsilon}\cup \tilde{F}^{\varepsilon},\\
+ \infty  \quad &\hbox{otherwise in } L^2(\Omega).
\end{cases}
\end{equation*}

\noindent
As first step, we localize the sequence $(\mathcal{F}^\varepsilon)$, introducing an explicit 
dependence on the set of integration. More explicitly, for every $u\in L^2(\Omega)$ and for every open set $A\in \mathcal{A}(\Omega)$ we define
\begin{equation*}
\mathcal{F}^\varepsilon(u,A):= 
\left\{
\begin{array}{ll}
\vspace{.2cm}
\displaystyle \int_{A}|\nabla u|^2 dx + \varepsilon\,\mathcal{H}^{n - 1}(S_{u}\cap A) \quad 
&\hbox{if } u\in SBV^2(A),\, S_u \subset \big(\tilde{E}^\varepsilon \cup \tilde{F}^{\varepsilon}\big)\cap A,\\
+ \infty \quad &\hbox{otherwise in } L^2(\Omega). 
\end{array}
\right.
\end{equation*}
For a fixed $u\in L^2(\Omega)$ we can extend the localized functional we have just defined to a measure $(\mathcal{F}^\varepsilon)^*(u,\cdot)$ 
on the class of Borel sets $\mathcal{B}(\Omega)$ in the usual way:
\begin{equation*}
(\mathcal{F}^\varepsilon)^*(u,B):= \inf \big\{\mathcal{F}^\varepsilon(u,A): A \in \mathcal{A}(\Omega), B \subseteq A\big\}. 
\end{equation*}

\subsection{Integral representation of the $\Gamma$-limit}

In this subsection we are going to prove that the sequence $(\mathcal{F}^\varepsilon)$ $\Gamma$-converges to a functional $\mathcal{F}^{hom}$, and that this limit functional admits an integral representation.
A preliminary result is given by next theorem, in which we prove the $\Gamma$-convergence of a suitable subsequence 
of $(\mathcal{F}^\varepsilon)$. 

\begin{thm}\label{Gc}
Let $\varepsilon$ be a sequence converging to zero. Then there exist a subsequence 
$(\sigma(\varepsilon))$ and a functional $\mathcal{F}^{hom}_\sigma: L^2(\Omega)\times \mathcal{A}(\Omega) 
\rightarrow [0,+ \infty]$ such that, for every $A\in \mathcal{A}(\Omega)$,
\begin{equation*}
\mathcal{F}^{hom}_\sigma(\cdot,A) = \Gamma\mbox{-}\lim_{\varepsilon \rightarrow 0} \mathcal{F}^{\sigma(\varepsilon)}(\cdot,A)
\end{equation*}
in the strong $L^2$-topology. Moreover, for every $u\in L^2(\Omega)$, the set function 
$\mathcal{F}^{hom}_\sigma(u,\cdot)$ is the restriction to $\mathcal{A}(\Omega)$ of a Borel measure on $\Omega$.
\end{thm}

Before giving the proof of this theorem, let us introduce some definitions and results that will be used in the following. 
For further references see \cite{DM93}.

\begin{defn}
Let $(G^\varepsilon)$ be a sequence of functionals on $L^2(\Omega)$. Define the functionals $G', G'': L^2(\Omega) \rightarrow \mathbb{R}$ as follows:
\begin{align*}
G':= \Gamma\mbox{-}\liminf_{\varepsilon \rightarrow 0} G^\varepsilon \quad \hbox{and}\quad
G'':= \Gamma\mbox{-}\limsup_{\varepsilon \rightarrow 0} G^\varepsilon.
\end{align*}
\end{defn}

\begin{defn}
We say that a functional $G: L^2(\Omega)\times \mathcal{A}(\Omega) \rightarrow [0, + \infty]$ is increasing 
(on $\mathcal{A}(\Omega)$) if for every $u \in L^2(\Omega)$ the set function $G(u,\cdot)$ is increasing on $\mathcal{A}(\Omega)$.
\end{defn}

\begin{defn}
Given a functional $G: L^2(\Omega)\times \mathcal{A}(\Omega) \rightarrow [0, + \infty]$, we define its inner regularization as
\begin{equation*}
G_{-}(u,A):= \sup \big\{G(u,B): B\in \mathcal{A}(\Omega), B\subset\subset A \big\}.
\end{equation*}
Observe that if $G$ is increasing, then also $G_{-}$ is increasing.
\end{defn}
 
\begin{defn}
We say that a sequence $G^\varepsilon$ is $\overline{\Gamma}$-convergent to a functional $G$ 
whenever 
$$G = (G')_{-} = (G'')_{-}.$$
\end{defn}

We have the following compactness theorem.
\begin{thm}\label{cpt}
Every sequence of increasing functionals has a $\overline{\Gamma}$-convergent subsequence.
\end{thm}

Next Theorem provides an extension of the fundamental estimate to $SBV^2$. The proof follows 
easily from \cite[Proposition 3.1]{BDfV}, but we will include the details for the convenience of the reader.

\begin{thm}[Fundamental estimate in $SBV^2$]\label{fe}
For every $\eta>0$ and for every $A', A''$ and $B$ $\in \mathcal{A}(\Omega)$, with $A'\subset\subset A''$, 
there exists a constant $M>0$ with the following property: for every $\varepsilon>0$ and for
every $u\in SBV^2(A'')$ such that $S_u \subset \big(\tilde{E}^\varepsilon \cup \tilde{F}^{\varepsilon}\big) \cap A''$, and for every $v\in SBV^2(B)$ such that $S_v \subset \big(\tilde{E}^\varepsilon \cup \tilde{F}^{\varepsilon}\big) \cap B$ there exists a function $\varphi \in C_0^\infty (\Omega)$ with $\varphi = 1$ in a neighbourhood of $\bar{A'}$, 
$spt\, \varphi \subset A''$ and $0\leq \varphi \leq 1$ such that
\begin{equation*}
\mathcal{F}^\varepsilon(\varphi\,u + (1 - \varphi)\,v, A'\cup B) \leq (1 + \eta)\,\mathcal{F}^\varepsilon(u,A'') + 
(1 + \eta)\,\mathcal{F}^\varepsilon(v,B) + M \int_{T} |u - v|^2 dx,
\end{equation*}
where $T:= (A''\setminus A')\cap B$.
\end{thm}

\begin{proof}
Let $\eta > 0$, $A'$, $A''$ and $B$ be as in the statement. Let $A_1,\dots A_{k + 1}$ be open 
subsets of $\mathbb{R}^n$ such that $A'\subset\subset A_1 \subset\subset A_2 
\subset\subset \dots \subset\subset A_{k + 1} \subset\subset A''$. 
For every $i = 1,\dots,k$ let $\varphi_i$ be a function in $C_0^{\infty}(\Omega)$ with $\varphi_i = 1$
on a neighborhood of $\bar{A}_i$ and $\hbox{spt}\,\varphi \subset A_{i+1}$.

Now, let $u$ and $v$ be as in the statement and define 
the function $w_i$ on $A'\cup B$ as $w_i:= \varphi_i u + (1 - \varphi_i)\,v$ (where $u$ and $v$ are arbitrarily 
extended outside $A''$ and $B$, respectively). 
For $i = 1,\dots,k$ set $T_i:= (A_{i+1}\setminus \bar{A}_i)\cap B$. We can write, 
for fixed $\varepsilon>0$,
\begin{align}\label{prima}
\mathcal{F}^\varepsilon(w_i, A'\cup B) &=\, \int_{A'\cup B}|\,\nabla w_i|^2 dx + 
\varepsilon\,\mathcal{H}^{n - 1}\big(S_{w_i}\cap (A'\cup B)\big)\nonumber\\
&=\,(\mathcal{F}^\varepsilon)^*(u, (A'\cup B)\cap\bar{A_i}) + (\mathcal{F}^\varepsilon)^*(v, B\setminus A_{i+1}) + \mathcal{F}^\varepsilon(w_i,T_i)\nonumber\\
&\leq \,\mathcal{F}^\varepsilon(u, A'') + \mathcal{F}^\varepsilon(v, B) + \mathcal{F}^\varepsilon(w_i,T_i).
\end{align}
We can write more explicitly the last term in the previous expression as
\begin{align}\label{seconda}
\mathcal{F}^\varepsilon(w_i,T_i) = &\, \int_{T_i} |\,\varphi_i \nabla u + (1 - \varphi_i)\,\nabla v 
+ \nabla\varphi_i(u - v)|^2 dx + \varepsilon\,\mathcal{H}^{n - 1}\big(S_{w_i}\cap T_i\big)\nonumber\\
\leq &\,\int_{T_i} |\,\varphi_i \nabla u + (1 - \varphi_i)\,\nabla v 
+ \nabla\varphi_i(u - v)|^2 dx + \varepsilon\,\mathcal{H}^{n - 1}\big(S_{u}\cap T_i\big) + 
\varepsilon\,\mathcal{H}^{n - 1}\big(S_{v}\cap T_i\big)\nonumber\\
=:&\, I^\varepsilon_i(T_i). 
\end{align}
We would like to control $I^\varepsilon_i(T_i)$ by means of $\mathcal{L}^n(T_i)$. Let us define 
$M_k:= \max_{1\leq i\leq k}||\nabla \varphi_i||^2_{L^{\infty}}$. Hence
\begin{align}\label{terza}
I^\varepsilon_i(T_i)\leq &\, 2\int_{T_i} |\,\varphi_i \nabla u + (1 - \varphi_i)\,\nabla v 
|^2 dx + 2\int_{T_i} |\,\nabla\varphi_i(u - v)|^2 dx \,+ \nonumber\\ 
&+\,\varepsilon\,\mathcal{H}^{n - 1}\big(S_{u}\cap T_i\big) + 
\varepsilon\,\mathcal{H}^{n - 1}\big(S_{v}\cap T_i\big)\nonumber\\
\leq &\, 2\int_{T_i} |\,\nabla u|^2 dx + 2\int_{T_i} |\,\nabla v|^2 dx + 
2\int_{T_i} |\,\nabla\varphi_i\,|^2 |u - v|^2 dx \,+ \nonumber\\ 
&+\,\varepsilon\,\mathcal{H}^{n - 1}\big(S_{u}\cap T_i\big) + 
\varepsilon\,\mathcal{H}^{n - 1}\big(S_{v}\cap T_i\big)\nonumber\\
\leq&\, 2\,\mathcal{F}^\varepsilon(u,T_i) + 2\,\mathcal{F}^\varepsilon(v,T_i) + 
2\,M_k \int_{T_i} |u - v|^2 dx = : J^\varepsilon(T_i).
\end{align}
Now, let $i_0 \in \{1,\dots,k\}$ be such that $T_{i_0}$ realizes $\min_{1\leq i \leq k}J^\varepsilon(T_i)$. 
Then, being $J^\varepsilon$ a measure, we have
\begin{equation}\label{quarta}
J^\varepsilon(T_{i_0}) \leq \frac{1}{k}\,\sum_{i = 1}^{k}J^\varepsilon(T_{i}) \leq \frac{1}{k}\,J^\varepsilon(T).
\end{equation}
Notice that $i_0 = i_0 (\varepsilon)$, it depends on $\varepsilon$. 

Combining together (\ref{prima})-(\ref{quarta}), we get
\begin{align}\label{bene}
\mathcal{F}^\varepsilon(w_{i_0}, A'\cup B) &\leq\, \mathcal{F}^\varepsilon(u,A'') + 
\mathcal{F}^\varepsilon(v, B) + \frac{1}{k}\,J^\varepsilon(T) \nonumber\\
=&\, \mathcal{F}^\varepsilon(u, A'') + 
\mathcal{F}^\varepsilon(v, B) + 
\frac{2}{k}\,\mathcal{F}^\varepsilon(u,T) + \frac{2}{k}\,\mathcal{F}^\varepsilon(v,T) + 
\frac{2}{k}\,M_k \int_{T} |u - v|^2 dx\nonumber\\
\leq&\, \mathcal{F}^\varepsilon(u,A'') + \mathcal{F}^\varepsilon(v, B) + 
\frac{2}{k}\,\mathcal{F}^\varepsilon(u,A'') + \frac{2}{k}\,\mathcal{F}^\varepsilon(v,B) + 
\frac{2}{k}\,M_k \int_{T} |u - v|^2 dx.
\end{align}
Now, since the choice of the number $k$ of the stripes between $A'$ and $A''$ is 
completely free, we can assume that $k$ is such that $\frac{2}{k} < \eta$. Hence
$k = k(\eta)$. Let us define $\overline{M}_\eta:= \frac{2}{k} M_k$; then in (\ref{bene})
we have
\begin{equation*}
\mathcal{F}^\varepsilon(w_{i_0}, A'\cup B) \leq (1 + \eta)\,\mathcal{F}^\varepsilon(u,A'') + 
(1 + \eta)\,\mathcal{F}^\varepsilon(v,B) + \overline{M}_\eta \int_{T} |u - v|^2 dx,
\end{equation*}
which is exactly the claim.
\end{proof}

Now we are ready to give the proof of Theorem \ref{Gc}.

\begin{proof} [Proof of Theorem \ref{Gc}]
Since for every $\varepsilon>0$ the functional $\mathcal{F}^\varepsilon$ is increasing, we deduce by Theorem \ref{cpt} 
that there exist a subsequence $(\sigma(\varepsilon))$ and a functional $\mathcal{F}_\sigma^{hom}: L^2(\Omega)
\times \mathcal{A}(\Omega) \rightarrow [0,+ \infty]$ such that $\mathcal{F}_\sigma^{hom} = \overline{\Gamma}(L^2)\mbox{-}\lim_{\varepsilon\rightarrow 0}\mathcal{F}^{\sigma(\varepsilon)}$. 
We put a subscript $\sigma$ in order to underline that the limit functional may depend on the subsequence. 
Now define the nonnegative increasing functional $J: L^2(\Omega)\times \mathcal{A}(\Omega) \rightarrow [0,+ \infty]$ as
\begin{equation*}
J(u,A):= 
\begin{cases}
\displaystyle \int_{A}|\nabla u|^2 dx  \quad &\hbox{if } u_{|A}\in H^1(A),\\
+ \infty \quad &\hbox{otherwise}.
\end{cases}
\end{equation*}
Clearly, $J$ is a measure with respect to $A$. Moreover $0\leq \mathcal{F}^{\sigma(\varepsilon)} \leq J$ for every $\varepsilon > 0$ and 
the fundamental estimate holds uniformly for the subsequence $(\mathcal{F}^{\sigma(\varepsilon)})$ by Theorem \ref{fe}. Then we can proceed as in \cite[Proposition 18.6]{DM93} and we obtain that  
\begin{equation*}
\mathcal{F}_\sigma^{hom}(u,A) = (\mathcal{F}_\sigma^{hom})'(u,A) = (\mathcal{F}^{hom}_\sigma)''(u,A)
\end{equation*}
for every $u\in L^2(\Omega)$ and for every $A\in \mathcal{A}(\Omega)$ such that 
$J(u,A) < +\infty$.

Fix $A \in \mathcal{A}(\Omega)$. As we noticed in Theorem \ref{liminf1}, we have the bound 
$\mathcal{F}^{\sigma(\varepsilon)}(\cdot,A) \geq \mathcal{G}^{\sigma(\varepsilon)}(\cdot,A)$, 
with $\mathcal{G}^{\sigma(\varepsilon)}$ defined in (\ref{defGe}). Hence by Theorem \ref{GammaG} the $\Gamma$-limit of 
$\mathcal{F}^{\sigma(\varepsilon)}(\cdot,A)$ is finite only on $H^1(A)$, which is the same domain 
where $J(\cdot,A)$ is finite, and is given by $\mathcal{F}^{hom}_\sigma(\cdot,A)$. This proves the 
stated convergence of a subsequence $\big(\mathcal{F}^{\sigma(\varepsilon)}\big)$.

Finally, $\mathcal{F}^\varepsilon(u,\cdot)$ is the restriction to $\mathcal{A}(\Omega)$ of a Borel measure 
on $\Omega$. Then, by Theorem \ref{fe} and \cite[Theorem 18.5]{DM93} we have that for every 
$u\in L^2(\Omega)$ the set function $\mathcal{F}^{hom}_\sigma(u,\cdot)$ is the restriction to $\mathcal{A}(\Omega)$ of 
a Borel measure on $\Omega$.
\end{proof}

Now we show some general properties for the $\Gamma$-limit of $\mathcal{F}^\varepsilon$, even if, up to
now, we have proved the convergence only for a subsequence. The fact that the whole sequence 
converges will follow from the characterization of the $\Gamma$-limit, which will depend only on the 
gradient of the displacement and not on the subsequence $\sigma(\varepsilon)$.
From now on let us assume that we have already proved it and postpone the proof to the end of the section. 
Hence we can omit the subscript $\sigma$ and call $\mathcal{F}^{hom}$ the $\Gamma$-limit of the whole sequence 
$(\mathcal{F}^\varepsilon)$.


\begin{lem}\label{citare}
The restriction of the functional $\mathcal{F}^{hom}: L^2(\Omega)\times \mathcal{A}(\Omega) \rightarrow [0, + \infty]$ 
to $H^1(\Omega)\times \mathcal{A}(\Omega)$ satisfies the following properties: 
for every $u,v \in H^1(\Omega)$ and for every $A\in \mathcal{A}(\Omega)$
\begin{itemize}
\item[(a)] $\mathcal{F}^{hom}$ is local, i.e., $\mathcal{F}^{hom}(u,A) = \mathcal{F}^{hom}(v,A)$ whenever $u_{|A} = v_{|A}$;
\vspace{.15cm}
\item[(b)] the set function $\mathcal{F}^{hom}(u,\cdot)$ is the restriction to $\mathcal{A}(\Omega)$ of a Borel measure on $\Omega$; 
\vspace{.15cm}
\item[(c)] $\mathcal{F}^{hom}(\cdot,A)$ is sequentially weakly lower semicontinuous on $H^1(\Omega)$;
\vspace{.15cm}
\item[(d)] for every $a\in \mathbb{R}$ we have $\mathcal{F}^{hom}(u,A) = 
\mathcal{F}^{hom}(u + a, A)$;
\vspace{.15cm}
\item[(e)] $\mathcal{F}^{hom}$ satisfies the bound 
\begin{equation*}
0 \leq \mathcal{F}^{hom}(u,A) \leq \int_{A}|D u|^2 dx.
\end{equation*}
\end{itemize}
\end{lem}
\begin{proof}
Properties (a) and (c) follow from the fact that $\mathcal{F}^{hom}(\cdot,A)$ is the $\Gamma$-limit of 
the sequence $\mathcal{F}^\varepsilon(\cdot,A)$, while (b) comes from Theorem \ref{Gc}. For property 
(d) we can proceed as follows. Let $u\in H^1(\Omega)$, $A\in\mathcal{A}(\Omega)$ and consider 
a recovery sequence $(u^\varepsilon) \subset L^2(\Omega) \cap SBV^2(A)$ satisfying the usual constraints for 
the jump set, converging to $u$ strongly in $L^2(\Omega)$ and such that $\big(\mathcal{F}^\varepsilon(u^\varepsilon,A)\big)$
converges to $\mathcal{F}^{hom}(u,A)$. Then $(u^\varepsilon + a)$ converges to $u + a$ in $L^2(\Omega)$ and 
\begin{equation*}
\mathcal{F}^{hom}(u + a,A)\leq \liminf_{\varepsilon\rightarrow 0}\mathcal{F}^\varepsilon(u^\varepsilon + a,A) = 
\liminf_{\varepsilon\rightarrow 0}\mathcal{F}^\varepsilon(u^\varepsilon,A) = \mathcal{F}^{hom}(u,A).
\end{equation*}
On the other hand, $\mathcal{F}^{hom}(u,A) = \mathcal{F}^{hom}((u+a) + (-a),A)\leq \mathcal{F}^{hom}(u + a,A)$, hence 
(d) is proved. For property (e), we just recall that the $\Gamma$-limit of the sequence $(\mathcal{F}^\varepsilon)$ is bounded from above by the Dirichlet functional, since that value is reached by a special sequence. 
\end{proof}

Next theorem shows that the functional $\mathcal{F}^{hom}$ admits an integral representation.

\begin{thm}
There exists a unique convex function $f : \mathbb{R}^n \rightarrow [0, + \infty[$ with the following properties: 
\begin{itemize}
\item[(i)] $0 \leq f(\xi) \leq |\xi|^2$ for every $\xi \in \mathbb{R}^n$;
\vspace{.15cm}
\item[(ii)] $ \displaystyle \mathcal{F}^{hom}(u,A) = \int_{A} f (D u)\,dx$ \, for every $A \in \mathcal{A}(\Omega)$ and 
for every $u\in H^1(A)$.
\end{itemize}
\end{thm}
\begin{proof}
Notice that the functional $\mathcal{F}^{hom}$ satisfies all the assumptions of \cite[Theorem 20.1]{DM93}, 
so thanks to Lemma \ref{citare} the Carath\'eodory function $f:\Omega\times \mathbb{R}^n \rightarrow \mathbb{R}$ defined as
\begin{equation}\label{integ}
f(y,\xi):= \limsup_{\varrho\rightarrow 0} \frac{\mathcal{F}^{hom}(\xi\cdot x, B_\varrho(y))}{\mathcal{L}^n(B_\varrho(y))}
\end{equation}
provides the integral representation 
\begin{equation*}
\mathcal{F}^{hom}(u,A) = \int_A f(x,D u)\,dx
\end{equation*}
for every $A\in \mathcal{A}(\Omega)$ and for every $u\in L^2(\Omega)$ such that $u_{|A} \in H^1(A)$. 
Moreover the same theorem ensures that for a.e. $x\in \Omega$ the function $f(x,\cdot)$ is convex 
on $\mathbb{R}^n$ and that 
\begin{equation*}
0 \leq f(x,\xi) \leq |\xi|^2 \quad \hbox{for a.e. } x\in \Omega \,\hbox{and for every } \xi \in \mathbb{R}^n.
\end{equation*}
It remains to show that $f$ is independent of the first variable. Using the definition (\ref{integ}), 
it is sufficient to prove that for every $y,z \in \Omega$ and $\xi\in \mathbb{R}^n$ and for every $\varrho > 0$, we have
\begin{equation}\label{ug}
\mathcal{F}^{hom}(\xi\cdot x, B_\varrho(y)) =  \mathcal{F}^{hom}(\xi\cdot x, B_\varrho(z)).
\end{equation}
Hence, let us fix $y,z \in \Omega$ and $\xi\in \mathbb{R}^n$ and $\varrho > 0$; being $\mathcal{F}^{hom}(\cdot,B_\varrho(y))$ a $\Gamma$-limit, there exists a recovery sequence $(u^\varepsilon) \subset SBV^2(B_\varrho(y))$ satisfying the usual constraint on the jump set, such that $u^\varepsilon \rightarrow 0$ strongly in $L^2(\Omega)$ and
\begin{equation*}
\lim_{\varepsilon\rightarrow 0} \mathcal{F}^\varepsilon(\xi\cdot x + u^\varepsilon,B_\varrho(y)) = 
\mathcal{F}^{hom}(\xi\cdot x, B_\varrho(y)).
\end{equation*}
Without loss of generality we can assume $(u^\varepsilon)\subset SBV^2_0(B_\varrho(y))$, where the 
subscript $0$ denotes the functions vanishing on the boundary. Indeed we can always reduce to this case by 
means of a cut-off function. Now let us define the vector $\tau^\varepsilon \in \mathbb{R}^n$ as 
\begin{equation*}
\tau^\varepsilon:= \varepsilon\,\bigg[\frac{z - y}{\varepsilon}\bigg],
\end{equation*}
where the symbol $[\cdot]$ denotes the integer part componentwise. Extend $u^\varepsilon$ by zero 
out of $B_\varrho(y)$ and define the new sequence $v^\varepsilon(x):= u^\varepsilon(x - \tau^\varepsilon)$. 
It turns out that $S_{v^\varepsilon}\subset \tilde{E}^\varepsilon \cup \tilde{F}^\varepsilon$; 
moreover $v^\varepsilon$ is identically zero out of $B_\varrho(y) + \tau^\varepsilon$ and it
converges to zero strongly in $L^2(\Omega)$. Observe that for small enough $\varepsilon$ and for every $r > 1$ 
we have that $B_\varrho(y) + \tau^\varepsilon \subset B_{r\varrho}(z)$. 
Hence the sequence 
$\xi\cdot x + v^\varepsilon$ gives a bound for $\mathcal{F}^{hom}(\xi\cdot x,B_{\varrho}(z))$, that is
\begin{align}\label{bfv}
\mathcal{F}^{hom}(\xi\cdot x,B_{\varrho}(z)) &\,\leq \mathcal{F}^{hom}(\xi\cdot x,B_{r\varrho}(z)) \leq 
\liminf_{\varepsilon\rightarrow 0} \mathcal{F}^\varepsilon(\xi\cdot x + v^\varepsilon, B_{r\varrho}(z))\nonumber\\
&=\, \liminf_{\varepsilon\rightarrow 0} \bigg\{\int_{B_{r\varrho}(z)}|\xi + \nabla v^\varepsilon|^2 dx + 
\varepsilon\,\mathcal{H}^{n - 1}(S_{v^\varepsilon} \cap B_{r\varrho}(z))\bigg\}.
\end{align}
We can rewrite the last line of (\ref{bfv}) in terms of $u^\varepsilon$, and so we get 
\begin{align*}
\mathcal{F}^{hom}(\xi\cdot x,B_{\varrho}(z)) &\leq\, 
\liminf_{\varepsilon\rightarrow 0} \bigg\{\int_{B_{\varrho}(y)}|\xi + \nabla u^\varepsilon|^2 dx + 
|\xi|^2 \mathcal{L}^n(B_{r\varrho}\setminus B_\varrho) + 
\varepsilon\,\mathcal{H}^{n - 1}(S_{u^\varepsilon} \cap B_{\varrho}(y))\bigg\}\\
&=\, \mathcal{F}^{hom}(\xi\cdot x,B_{\varrho}(y)) + |\xi|^2 \mathcal{L}^n(B_{r\varrho}\setminus B_\varrho).
\end{align*}
Now, if we let $r \rightarrow 1$ we have that $\mathcal{F}^{hom}(\xi\cdot x,B_{\varrho}(z))\leq \mathcal{F}^{hom}(\xi\cdot x,B_{\varrho}(y))$. 
The reverse inequality can be deduced in the same way, hence the claim follows.
\end{proof}

\subsection{Homogenization formula}
Once we have shown that the $\Gamma$-limit of the sequence $(\mathcal{F}^\varepsilon)$ admits 
an integral representation, it remains to characterize the limit density. We will prove that 
it solves an asymptotic cell problem.

We define the function $f_{hom}: \mathbb{R}^n \rightarrow [0,+\infty)$ as
\begin{equation}\label{deff0}
f_{hom}(\xi):= \lim_{t\rightarrow +\infty}\frac{1}{t^n}\,\inf\bigg\{
\int_{(0,t)^n} |\xi + \nabla w|^2 d\,x + \mathcal{H}^{n - 1}(S_w) : w\in SBV^2_0\big((0,t)^n\big), S_w \subset \tilde{E}\cup\tilde{F} \bigg\} 
\end{equation}
where, according to the notation used so far, we have
\begin{equation*}
\tilde{E}:= (E + \mathbb{Z}^n), \quad \tilde{F}:= (F + \mathbb{Z}^n).
\end{equation*}

\begin{thm}
The function $f_{hom}$ in (\ref{deff0}) is well defined, that is the function
\begin{equation}\label{defg}
g(t):= \frac{1}{t^n}\,\inf\bigg\{
\int_{(0,t)^n} |\xi + \nabla w|^2 dx + \mathcal{H}^{n - 1}(S_w) : w\in SBV^2_0\big((0,t)^n\big), S_w \subset \tilde{E}\cup \tilde{F}\bigg\}
\end{equation}
admits a limit as $t \rightarrow + \infty$.
\end{thm}

\begin{proof}
Let $\xi \in\mathbb{R}^n$ and let $t>0$; by definition of $g$, there exists a function 
$u_t \in SBV^2_0\big((0,t)^n\big)$ with $S_{u_t} \subset \tilde{E}\cup\tilde{F}$ such that
\begin{equation*}
\frac{1}{t^n}\, \bigg\{\int_{(0,t)^n} |\xi + \nabla u_t|^2 d\,x + \mathcal{H}^{n - 1}(S_{u_t})\bigg\}
\leq g(t) + \frac{1}{t}.
\end{equation*}
Fix $s>t$ and define a subset of $\mathbb{N}^n$ as
\begin{equation*}
K:= \big\{ \mathbf{k} = (k_1,\dots,k_n) \in \mathbb{N}^n : 0< ([t] + 1)\,k_j < s,\,\hbox{for } j=1,\dots,n\big\}.
\end{equation*}
Then, we define the set $I:= ([t] + 1)K$.
Now, consider the function $u_s: \mathbb{R}^n \rightarrow \mathbb{R}$ defined in the following way:
\begin{equation*}
u_s(x):= 
\begin{cases}
u_t(x - \mathbf{i}) & \hbox{if } x \in \mathbf{i} + (0,t)^n, \mathbf{i} \in I,\\
0 & \hbox{otherwise.}
\end{cases}
\end{equation*}
The fact that we performed a translation by integers and the $Q$-periodicity of the jumps for the 
function $u_t$ entail
$S_{u_s}\subset \tilde{E}\cup\tilde{F}$. Moreover, $u_s$ vanishes on the boundary of 
$(0,s)^n$. Hence, $u_s$ is a competitor for $g(s)$, and so
\begin{equation*}
g(s)\leq  \frac{1}{s^n}\,\bigg\{
\int_{(0,s)^n} |\xi + \nabla u_s|^2 d\,x + \mathcal{H}^{n - 1}(S_{u_s}) 
\bigg\}.
\end{equation*}
Define the set $R^s_t \subset (0,s)^n$ as
\begin{equation*}
R^s_t:= (0,s)^n \setminus \bigcup_{\mathbf{i}\in I} \big(\mathbf{i} + (0,t)^n\big). 
\end{equation*}
Since for the cardinality of the set $I$ we have
\begin{equation}\label{cardI}
\frac{s^n}{([t] + 1)^n} - 1 < |I| = \Big(\Big[\frac{s}{[t] + 1}\Big]\Big)^n \leq \frac{s^n}{([t] + 1)^n},
\end{equation}
then it turns out that 
\begin{equation}\label{resto}
\mathcal{L}^n(R^s_t) = s^n - \Big(\Big[\frac{s}{[t] + 1}\Big]\Big)^n t^n
\leq s^n - \Big(\frac{s - ([t] + 1)}{[t] + 1}\Big)^n t^n.
\end{equation}
Notice that $u_s = 0$ on $R^s_t$ and that $S_{u_s} \cap R^s_t = \emptyset$; therefore
\begin{align*}
g(s)&\,\leq \frac{1}{s^n}\,\bigg\{ \mathcal{L}^n(R^s_t)\,|\xi|^2 + \sum_{\mathbf{i}\in I} 
\int_{\mathbf{i} + (0,t)^n} |\xi + \nabla u_s|^2 d\,x  + \sum_{\mathbf{i}\in I} 
\mathcal{H}^{n - 1}\big(S_{u_s}\cap(\mathbf{i} + (0,t)^n)\big)\bigg\}\\
&=\,\frac{1}{s^n}\,\bigg\{ \mathcal{L}^n(R^s_t)\,|\xi|^2 + \sum_{\mathbf{i}\in I} 
\int_{(0,t)^n} |\xi + \nabla u_t|^2 d\,x  + \sum_{\mathbf{i}\in I} 
\mathcal{H}^{n - 1}\big(S_{u_t}\cap(0,t)^n\big)\bigg\}.
\end{align*}
Using (\ref{cardI}) and (\ref{resto}) we obtain, finally,
\begin{equation*}
g(s) \leq \,\frac{t^n}{([t] + 1)^n}\,\Big(g(t) + \frac{1}{t}\,\Big) + 
\,|\xi|^2 \bigg(1 - \Big(\frac{s - t - 1}{s}\Big)^n\Big(\frac{t}{t + 1}\Big)^n\bigg).
\end{equation*}
Taking first the upper limit as $s \rightarrow + \infty$ and then the lower limit as 
$t \rightarrow + \infty$ we get
\begin{equation*}
\limsup_{s \rightarrow + \infty} g(s) \leq \liminf_{t \rightarrow + \infty} g(t),
\end{equation*}
and this concludes the proof.
\end{proof}

\noindent
Next theorem shows that the $\Gamma$-limit of the sequence $(\mathcal{F}^\varepsilon)$ can be expressed in 
terms of the homogenization formula (\ref{deff0}).

\begin{thm}
The function $f$ appearing in the expression of the limit functional $\mathcal{F}^{hom}$ and the function $f_{hom}$ 
defined by the asymptotic cell problem coincide, i.e., for every $\xi \in \mathbb{R}^n$ it turns out that
\begin{equation*}
f(\xi) = f_{hom}(\xi).
\end{equation*}
\end{thm}
\begin{proof}

\textit{First step: $f\geq f_{hom}$}.

Let $\xi\in\mathbb{R}^n$ and define $u_\xi(x):= \xi\cdot x$ for every $x\in \mathbb{R}^n$. By definition of $\Gamma$-convergence, there exists a recovery sequence $u^\varepsilon\subset SBV^2(Q)$ with 
$S_{u^\varepsilon}\subset \big(\tilde{E}^{\varepsilon} \cup \tilde{F}^{\varepsilon}\big)\cap Q$, 
such that $u^\varepsilon\rightarrow u_\xi$ strongly in $L^2(Q)$ and
\begin{equation*}
\lim_{\varepsilon \rightarrow 0}\mathcal{F}^\varepsilon(u^\varepsilon,Q) = \mathcal{F}^{hom}(u_\xi, Q) = f(\xi). 
\end{equation*}
Let us write $u^\varepsilon =: u_\xi + v^\varepsilon$, where $v^\varepsilon\subset SBV^2(Q)$ and 
$v^\varepsilon\rightarrow 0$ strongly in $L^2(Q)$.  
Without loss of generality we can assume $v^\varepsilon \in SBV^2_0(Q)$. Hence
\begin{equation}\label{cell}
f(\xi) = \lim_{\varepsilon \rightarrow 0}\mathcal{F}^\varepsilon(u_\xi + v^\varepsilon,Q) = 
\lim_{\varepsilon \rightarrow 0} \bigg\{\int_Q |\xi + \nabla v^\varepsilon|^2 d\,x + 
\varepsilon\,\mathcal{H}^{n - 1}(S_{v^\varepsilon})\bigg\}.
\end{equation}
Now, let us define the function $w^\varepsilon \in SBV^2_0(Q/\varepsilon)$ as
\begin{equation*}
v^\varepsilon(x)=: \varepsilon\,w^\varepsilon\Big(\frac{x}{\varepsilon}\Big).
\end{equation*}
Remark that $S_{w^\varepsilon} \subset \tilde{E}\cup \tilde{F}$. Then, rewriting (\ref{cell}) 
in terms of $w^\varepsilon$ we obtain
\begin{align*}
f(\xi) =& \,\lim_{\varepsilon \rightarrow 0} \varepsilon^n\bigg\{\int_{Q/\varepsilon} 
|\xi + \nabla w^\varepsilon|^2 dx + 
\mathcal{H}^{n - 1}(S_{w^\varepsilon})\bigg\}\\
\geq& \,\lim_{\varepsilon \rightarrow 0} \varepsilon^n
\inf \bigg\{\int_{(0,\frac{1}{\varepsilon})^n} |\xi + \nabla w|^2 dx + 
\mathcal{H}^{n - 1}(S_w): w\in SBV^2_0\big(\big(0,1/\varepsilon\big)^n\big),\,S_w \subset \tilde{E}\cup \tilde{F}\bigg\}\\
=&\, f_{hom}(\xi).
\end{align*}

\textit{Second step: $f\leq f_{hom}$}.

\noindent
Let $\xi\in\mathbb{R}^n$ and $l\in\mathbb{N}$; consider a function 
$w \in SBV^2_0((0,l)^n)$, with $S_w\subset \tilde{E}\cup \tilde{F}$, such that
\begin{align}\label{infw}
&\int_{(0,l)^n}|\xi + \nabla w|^2 dx + \mathcal{H}^{n - 1}(S_{w}) \nonumber\\
\leq&\, \inf \bigg\{ \int_{(0,l)^n}|\xi + \nabla v|^2 dx + \mathcal{H}^{n - 1}(S_{v}) :
v \in SBV^2_0((0,l)^n), S_v\subset \tilde{E}\cup \tilde{F} \bigg\} + 1.
\end{align}

\noindent
Let us define the sequence $u^\varepsilon : Q \rightarrow \mathbb{R}$ as 
\begin{equation*}
u^\varepsilon(x):= \xi\cdot x + \varepsilon\,\tilde{w}\Big(\frac{x}{\varepsilon}\Big),  
\end{equation*}
where $\tilde{w}$ denotes the function defined in the whole $\mathbb{R}^n$, obtained 
through a periodic extension of $w$. We have that $\mathcal{F}^\varepsilon(u^\varepsilon,Q) < +\infty$, 
being $S_{u^\varepsilon} \subset \tilde{E}^{\varepsilon}
\cup \tilde{F}^{\varepsilon}$, and that $u^\varepsilon$ converges to $\xi\cdot x$ strongly in $L^2(Q)$. 
Moreover
\begin{equation*}
\mathcal{F}^\varepsilon(u^\varepsilon,Q) = \int_{Q}|\nabla u^\varepsilon|^2 dx + \varepsilon\,
\mathcal{H}^{n - 1}(S_{u^\varepsilon}) = \varepsilon^n \bigg\{\int_{Q/\varepsilon}|\xi + 
\nabla \tilde{w}|^2 dx + \mathcal{H}^{n - 1}(S_{\tilde{w}})\bigg\}.
\end{equation*}
Now, in order to use the periodicity of $\tilde{w}$, we can write the domain $Q/\varepsilon$ as
union of (suitably translated) periodicity cells $(0,l)^n$. Assume for simplicity that 
$Q/\varepsilon$ is covered exactly by an integer number of these cells, that is by 
$1/(l\,\varepsilon)^n$ cells. 
Indeed, the integral over the remaining part of $Q/\varepsilon$ is a term of order $1/(l\,\varepsilon)^{n-1}$ 

Using (\ref{infw}), we get
\begin{align*}
&\mathcal{F}^\varepsilon(u^\varepsilon,Q) = \,\frac{1}{l^n}\bigg\{\int_{(0,l)^n} |\xi + 
\nabla w|^2 dx + \mathcal{H}^{n - 1}(S_{w})\bigg\} \nonumber\\
&\leq \,\frac{1}{l^n}\,\inf \bigg\{ \int_{(0,l)^n}|\xi + \nabla v|^2 dx + 
\mathcal{H}^{n - 1}(S_{v}) :
v \in SBV^2_0((0,l)^n), S_w\subset \tilde{E}\cup\tilde{F} \bigg\} + \frac{1}{l^n}.
\end{align*}
Taking first the $\limsup$ of both sides as $\varepsilon\rightarrow 0$ and then letting $l \rightarrow + \infty$ we obtain
\begin{equation*}
\limsup_{\varepsilon \rightarrow 0} \mathcal{F}^\varepsilon(u^\varepsilon,Q) \leq f_{hom}(\xi),
\end{equation*}
hence the claim is proved.
\end{proof} 

\noindent
Notice that from this theorem we deduce that the whole sequence $(\mathcal{F}^\varepsilon)$ 
$\Gamma$-converges, since the formula for the limit energy density does not depend on the subsequence.


\noindent
Up to now we have proved that the $\Gamma$-limit of the sequence
$\mathcal{F}^\varepsilon$ can be expressed through an asymptotic cell problem.
Nevertheless it is desirable to give a more explicit description of the 
density $f_{hom}$ and this will be partially done in the next lemmas.

\begin{lem}\label{notquad}
The functional $\mathcal{F}^{hom}$ is not a quadratic form.
\end{lem}
\begin{proof}
\textit{First step.} For every $\xi \in \mathbb{R}^n$ the following estimate holds:
\begin{equation}\label{est}
 A_0 \xi \cdot \xi \leq f_{hom}(\xi) \leq A_0 \xi \cdot \xi + P(E,Q),
\end{equation}
where $P(E,Q)$ denotes the perimeter of the set $E$ in $Q$, according to the notation 
introduced in Section 2.

Indeed, the lower bound follows from (\ref{impbound}) and Remark \ref{remark}. 
For the upper bound, by the definition of $\Gamma$-limit it is sufficient to find a sequence 
$u^\varepsilon \subset SBV^2(\Omega)$ with $S_{u^\varepsilon} \subset \tilde{E}^\varepsilon
\cup\tilde{F}^\varepsilon$ and converging to $u_{\xi}:= \xi\cdot x$ strongly in $L^2(\Omega)$, 
such that
\begin{equation*}
 \lim_{\varepsilon \rightarrow 0}\mathcal{F}^\varepsilon (u^\varepsilon) = A_0 \xi \cdot \xi + P(E,Q).
\end{equation*}
To this aim, we just take as $u^\varepsilon$ the recovery sequence introduced
in the proof of Theorem \ref{bfab}. 

\noindent
\textit{Second step.} For every $\xi \in \mathbb{R}^n \setminus \{0\}$, we have
\begin{equation}\label{diver}
A_0 \xi \cdot \xi \lneqq |\xi|^2.
\end{equation}
Indeed, for $\xi \neq 0$, we have
\begin{align*}
A_{0}\xi\cdot\xi &=\, \min \Big\{\int_{Q\setminus E}|\,\xi + \nabla w(y)|^2 dy : 
w\in SBV^2_{\#} (Q), S_{w} \subset E\cup F\Big\}\\
& \leq\,  \int_{Q\setminus E}|\xi|^2 dy = \mathcal{L}^n(Q\setminus E)\,|\xi|^2 < |\xi|^2,
\end{align*}
since $0< \mathcal{L}^n(Q\setminus E) < \mathcal{L}^n(Q) = 1$.

\noindent
\textit{Third step.}
For every $\xi \in \mathbb{R}^n \setminus \{0\}$ we have 
\begin{equation}\label{div}
f_{hom}(\xi) \gneqq A_0 \xi \cdot \xi.
\end{equation}
To prove (\ref{div}) it is enough to show that, for every $\xi\neq 0$ and for every 
admissible sequence $u^\varepsilon$ converging to $u_\xi = \xi\cdot x$ strongly in $L^2(\Omega)$, we have
\begin{equation}\label{cl}
\limsup_{\varepsilon \rightarrow 0} \mathcal{F}^\varepsilon (u^\varepsilon) > \mathcal{L}^n(\Omega)\, A_0 \xi \cdot \xi.
\end{equation}
We can restrict to the case $\mathcal{F}^\varepsilon (u^\varepsilon) < + \infty$, otherwise there is nothing to prove.
For the sake of simplicity, let us assume that $\Omega = Q$. We will treat separately the case in which
$u^\varepsilon$ has no jumps and the general case.

\textit{Case \,$S_{u^\varepsilon} = \emptyset$ for every $\varepsilon > 0$.}
Being $\mathcal{F}^\varepsilon (u^\varepsilon) = \int_{Q}|\nabla u^\varepsilon|^2 d\,x < + \infty$, we have 
that the sequence $(u^\varepsilon)$ is bounded in $H^1(Q)$. In particular this implies that 
$\nabla u^\varepsilon \rightharpoonup \xi$ weakly in $L^2(Q)$. By the weakly lower semicontinuity of the 
Dirichlet integral we deduce that
\begin{equation*}
|\xi|^2 \leq \liminf_{\varepsilon\rightarrow 0} \mathcal{F}^\varepsilon(u^\varepsilon), 
\end{equation*}
which together with (\ref{diver}), gives (\ref{cl}).

\textit{Case \,$S_{u^\varepsilon} \neq \emptyset$ for some $\varepsilon > 0$.} 
Let us fix $\beta>0$ independent of $\varepsilon$ and classify the cubes $Q_k^{\varepsilon}$ according to 
$\mathcal{H}^{n - 1} (S_{u^\varepsilon}\cap Q_k^{\varepsilon})$ being smaller or larger than 
$\beta\,\varepsilon^{n - 1}$. From what we proved in Theorem \ref{bfbw}, it is possible to choose 
the parameter $\beta$ in such a way that the cubes where 
$\mathcal{H}^{n - 1} (S_{u^\varepsilon}\cap Q_k^{\varepsilon}) \leq \beta\,\varepsilon^{n - 1}$ 
can be assumed to be undamaged. 

Hence we can divide the cubes $Q_k^{\varepsilon}$ in two classes: the undamaged cubes and the ones such that $\mathcal{H}^{n - 1} (S_{u^\varepsilon}\cap Q_k^{\varepsilon}) > \beta\,\varepsilon^{n - 1}$, where $\beta>0$ is a small constant, independent of $\varepsilon$. Denote by $N_d(\varepsilon)$ the number of damaged cubes. 
From the expression of the functional no bound for $N_d(\varepsilon)$ can be derived, i.e., it may happen that 
$\mathcal{H}^{n - 1}(S_{u^\varepsilon}\cap Q_k^{\varepsilon}) > \beta\,\varepsilon^{n - 1}$ for every $k =1.\dots,N(\varepsilon)$. In any case it is clear that $\varepsilon^n N_d(\varepsilon)$ is a bounded quantity. 
According to the behaviour of $N_d(\varepsilon)$ as $\varepsilon\rightarrow 0$, three different cases may arise.

1) Assume that the number of damaged cube is small, that is 
\begin{equation}\label{few}
\limsup_{\varepsilon \rightarrow 0} \varepsilon^n N_d(\varepsilon) = 0.
\end{equation} 
Define the function $a^\varepsilon: Q \rightarrow \mathbb{R}$ as
\begin{equation*}
a^\varepsilon(x) := 
\begin{cases}
0 \quad & \hbox{in the damaged } Q_{k}^{\varepsilon},\\
1 \quad & \hbox{otherwise in } Q.
\end{cases}
\end{equation*}
From (\ref{few}) we have that $a^\varepsilon \rightarrow 1$ strongly in $L^1(Q)$. Now,
\begin{align*}
\mathcal{F}^{\varepsilon}(u^\varepsilon) &=\, \int_{Q}|\nabla u^\varepsilon|^2 dx + 
\varepsilon\,\mathcal{H}^{n - 1}(S_{u^\varepsilon}) \\
&\geq\, \int_{Q}a^\varepsilon(x)\,|\nabla u^\varepsilon|^2 dx + 
\beta\,\varepsilon^n N_d(\varepsilon).
\end{align*}
Then, taking the $\liminf$ as $\varepsilon \rightarrow 0$ we get
\begin{equation*}
\liminf_{\varepsilon \rightarrow 0}\mathcal{F}^{\varepsilon}(u^\varepsilon) \geq |\xi|^2,
\end{equation*}
so also in this case (\ref{cl}) follows from (\ref{diver}).

2) Assume that the number of damaged cube is high, that is 
\begin{equation}\label{many}
\liminf_{\varepsilon \rightarrow 0} \varepsilon^n N_d(\varepsilon) = C > 0.
\end{equation} 
In this case we can say that, for $\varepsilon$ small enough, we have $\varepsilon^n N_d(\varepsilon) > C/2$. 
Hence, recalling the definition (\ref{defGe}) after a suitable extension of $u^\varepsilon$ in $\tilde{E}^\varepsilon$, we have
\begin{equation*}
\mathcal{F}^{\varepsilon}(u^\varepsilon) =\, \int_{Q}|\nabla u^\varepsilon|^2 dx + 
\varepsilon\,\mathcal{H}^{n - 1}(S_{u^\varepsilon}) 
\geq\, \mathcal{G}^\varepsilon(u^\varepsilon) + \beta\,\varepsilon^n N_d(\varepsilon) 
\geq\, \mathcal{G}^\varepsilon(u^\varepsilon) + \beta\,\frac{C}{2}.
\end{equation*}
Then, taking the $\liminf$ as $\varepsilon \rightarrow 0$ we get by Theorem \ref{GammaG}
\begin{equation*}
\liminf_{\varepsilon \rightarrow 0}\mathcal{F}^{\varepsilon}(u^\varepsilon) \geq A_0\xi\cdot \xi + \beta\,\frac{C}{2},
\end{equation*}
so also in this case (\ref{cl}) holds.

3) Finally, let us analyze the intermediate case. Assume that 
\begin{equation*}
\liminf_{\varepsilon \rightarrow 0} \varepsilon^n N_d(\varepsilon) = 0.
\end{equation*}
and
\begin{equation*}
\limsup_{\varepsilon \rightarrow 0} \varepsilon^n N_d(\varepsilon) = C > 0.
\end{equation*}
Consider a subsequence $\varepsilon_k$ such that 
\begin{equation*}
\lim_{k \rightarrow \infty} \varepsilon_k^n N_d(\varepsilon_k) = 
\limsup_{\varepsilon \rightarrow 0} \varepsilon^n N_d(\varepsilon).
\end{equation*}
Then, we can apply the result of the previous case to this subsequence and we get 
\begin{equation*}
\limsup_{k \rightarrow \infty}\mathcal{F}^{\varepsilon_k}(u^{\varepsilon_k}) \geq A_0\xi\cdot \xi + \beta\,\frac{C}{2}.
\end{equation*}
Being the $\limsup$ of the whole sequence bigger or equal to the $\limsup$ of a subsequence, we have 
the thesis (\ref{cl}).

\noindent
\textit{Fourth step.}
Assume by contradiction that $f_{hom}$ is $2$-homogeneous. 
Hence replacing $\xi$ with $\lambda\,\xi$ in (\ref{est}) we have that, for every $\lambda\in\mathbb{R}$,
\begin{equation}\label{estlambda}
\lambda^2 A_0 \xi \cdot \xi \leq \lambda^2 f_{hom}(\xi) \leq \lambda^2 A_0 \xi \cdot \xi + P(E,Q).
\end{equation}
Dividing by $\lambda^2$ and letting $\lambda \rightarrow + \infty$ one gets
\begin{equation*}
f_{hom}(\xi) = A_0 \xi \cdot \xi,
\end{equation*}
which is in contrast with (\ref{div}). This shows that $f_{hom}$ is not $2$-homogeneous and therefore 
$\mathcal{F}^{hom}$ is not a quadratic form.
\end{proof}

\begin{rem}
The estimates (\ref{est}) and (\ref{div}) proved in the previous lemma can be summarized by the formula
\begin{equation}\label{costr}
 A_0 \xi\cdot\xi \lneqq f_{hom}(\xi) \leq \min \big\{|\xi|^2, A_0 \xi \cdot \xi + P(E,Q)\big\},
\end{equation}
that holds true for every $\xi \in \mathbb{R}^n \setminus \{0\}$.

It is clear that there exists a threshold $M>0$ such that 
\begin{equation}\label{asy}
A_0 \xi \cdot \xi + P(E,Q) \lneqq  |\xi|^2 \quad \mbox{for every } |\xi| > M.
\end{equation}
Condition (\ref{asy}) together with (\ref{costr}) entail in particular that 
\begin{equation*}
 f_{hom}(\xi) \lneqq |\xi|^2 \quad \mbox{for every } |\xi|> M,
\end{equation*}
that is, for $|\xi|$ sufficiently big, the limit density is strictly smaller than $|\xi|^2$. 

The situation is clarified by the following figure

\vspace{-.2 cm}
\begin{figure}[htbp]
\begin{center}
\includegraphics[height=.3\textwidth]{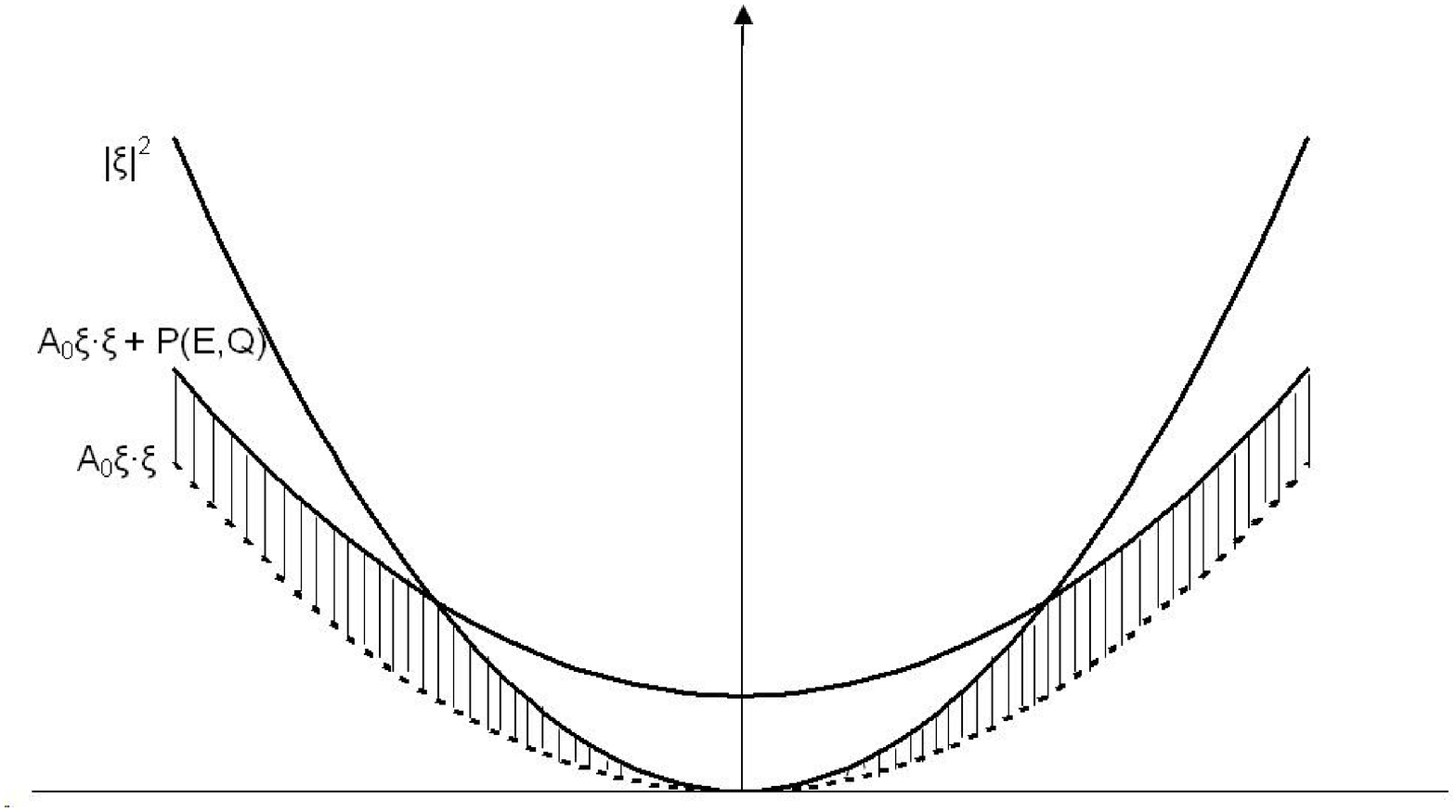}
\end{center}
\caption{Limit energy density}
\label{visco}
\end{figure}

%

It is not yet clear the behaviour of $f_{hom}(\xi)$ for $|\xi|$ very small, but we 
expect that 
\begin{equation*}
 \lim_{|\xi|\rightarrow 0} \frac{f_{hom}(\xi)}{|\xi|^2} = 1.
\end{equation*}
\end{rem}

\vspace{.2cm}
\noindent
Lemma \ref{notquad} shows also that the functional $\mathcal{F}^{hom}$ is not a quadratic form and it is 
not even $2$-homogeneous. Next lemma clarifies how $2$-homogeneity is violated.

\begin{lem}
For every $\xi\in\mathbb{R}^n$ and every $\lambda \geq 1$ we have the inequality
\begin{equation}\label{leq}
f_{hom}(\lambda\,\xi) \leq \lambda^2 f_{hom}(\xi),
\end{equation}
while for every $\xi\in\mathbb{R}^n$ and every $0 < \lambda \leq 1$ we have the reverse inequality
\begin{equation}\label{geq}
f_{hom}(\lambda\,\xi) \geq \lambda^2 f_{hom}(\xi).
\end{equation}
\end{lem}

\begin{proof}
Let $\xi\in\mathbb{R}^n$ be given and let $w\in SBV_0^2((0,t)^n)$ with $S_w \subset \tilde{E}\cup\tilde{F}$. 
Consider $\lambda \geq 1$ and set $w_\lambda:= \lambda \, w$. Clearly it turns out that 
$w_\lambda\in SBV_0^2((0,t)^n)$ and $S_{w_\lambda} \subset \tilde{E}\cup\tilde{F}$. Moreover

\begin{equation}\label{bu}
\int_{(0,t)^n}|\xi + \nabla w|^2 dx + \mathcal{H}^{n - 1}(S_{w}) 
\geq \,\frac{1}{\lambda^2}\,\bigg\{ \int_{(0,t)^n}|\,\lambda\,\xi + \nabla w_\lambda|^2 dx + \mathcal{H}^{n - 1}\big(S_{w_\lambda}\big)\bigg\}.
\end{equation}
Now, if we take the infimum of both sides of (\ref{bu}) over all $w\in SBV_0^2((0,t)^n)$ with $S_w \subset \tilde{E}\cup\tilde{F}$, we divide by $t^n$ the resulting expression and let $t\rightarrow +\infty$, we obtain 
exactly (\ref{leq}), using the definition (\ref{deff0}). 

Proceeding in a similar way we get the reverse inequality (\ref{geq}) in the case $\lambda\leq 1$. 
\end{proof}


\section{Appendix}
\noindent
In this appendix we present an alternative proof of Theorem \ref{bfbw} in the case of 
a two-dimensional domain $\Omega$. This proof is based on the maximum principle, which 
allows us to estimate the local opening of the crack in a small ball surrounding the crack. 
It is therefore strictly bidimensional. 
A similar method can be found in \cite{CGP05} and in \cite{DMMS}.

We use the same notation as in the previous sections. In particular we denote with $Q:= (0,1)^2$ the unit cube 
and with $Q_{\delta} \subset \subset Q_{\hat{\delta}} \subset\subset Q$ the concentric cubes with distance 
$\delta$ and $\hat{\delta}$ from $\partial Q$, respectively. Let $E, F \subset Q_\delta$ be the sets where
a crack may appear, satisfying the assumptions required in Section 2. Let us fix a 
boundary displacement on $\partial Q_{\hat{\delta}}$, given by the trace of a function $\varphi\in H^1(Q)$,
and let $0< \beta< (\delta - \hat{\delta})/2$ be a parameter.

Let $\tilde{v}$ be the elastic solution corresponding to the datum $\varphi$, that is the solution to the problem  
\begin{equation*}
\hbox{(Dir)} \quad  \displaystyle \min\bigg\{\int_{Q_{\hat{\delta}}} |\nabla w|^2 dx : w\in H^1(Q_{\hat{\delta}}), 
w = \varphi\, \hbox{on } \partial Q_{\hat{\delta}}\bigg\},
\end{equation*}
and let $\hat{v}$ be a solution to the problem 
\begin{equation*}
\begin{array}{ll}
&\hbox{(MS)} \quad \displaystyle \min\bigg\{\int_{Q_{\hat{\delta}}} |\nabla w|^2 dx + 
\mathcal{H}^{1}(S_{w}) : w \in SBV^2(Q_{\hat{\delta}}), S_{w}\subset E\cup F,\\
& \hspace{6.5cm}\mathcal{H}^{1}(S_{w}) \leq \beta, w = \varphi\, \hbox{on } \partial Q_{\hat{\delta}}\bigg\}.
\end{array}
\end{equation*}
The main result of this section is the following.

\begin{thm}\label{app}
For every $\beta$ small enough, there exists a constant 
$\omega(\beta)>0$ with $\omega(\beta)\rightarrow 0$ as $\beta\rightarrow 0$ such that the 
functions $\tilde{v}$ and $\hat{v}$ defined by the problems (Dir) 
and (MS), respectively, satisfy
the following relation:
\begin{equation}\label{trick}
\int_{Q_{\hat{\delta}}} |\nabla \hat{v}|^2 dx + \mathcal{H}^{1}(S_{\hat{v}}) 
\geq (1 - \omega(\beta))\int_{Q_{\hat{\delta}}}|\nabla \tilde{v}|^2 dx. 
\end{equation}
\end{thm}

\begin{rem}
Theorem \ref{app} ensures that if a function has a ``small" jump set, then it can be replaced with a 
function which has no discontinuities, up to a ``small" error in terms of the energy, 
depending on the measure of the jump set. 

This is exactly what we proved in (\ref{veroclaim}) within Theorem \ref{bfbw}. 
As we have already noticed, the proof of Theorem \ref{app} works only in dimension $2$, 
but it has the advantage of being more direct.
\end{rem}

\begin{proof}[Proof of Theorem \ref{app}] 
Let $\hat{v}$ be a minimizer for the problem (MS) and let us set
\begin{equation}\label{jump}
\Gamma:= S_{\hat{v}}.
\end{equation}

We notice that we can arbitrarily change the (constant) values of the function $\hat{v}$ in the regions 
where the gradient is zero, and the resulting function is still a minimizer for the same problem. 
So our first step is to fix the constants in these regions. 

\textit{Properties of $\Gamma$}. We shall split $\Gamma$ in two parts, called  
$\Gamma_*$ and $\Gamma \setminus \Gamma_*$, where $\Gamma_*$ will be related to the 
sets on which $\hat{v}$ is constant.

Let $G\subset Q_{\hat{\delta}}$ be a set having finite perimeter in $Q_{\hat{\delta}}$, 
maximal with respect to inclusion, such that $\partial^* G\subset \Gamma$. Assume that 
$\mathcal{L}^2(G) > 0$. 

It is easy to show that the function $\hat{v}$ is constant in $G$. In fact otherwise 
we can define, for a constant $c\in\mathbb{R}$, the function
\begin{equation*}
w:= 
 \begin{cases}
  \hat{v} & \hbox{in } Q_{\hat{\delta}}\setminus G,\\
  c & \hbox{in } G.
 \end{cases}
\end{equation*}
It turns out that $w$ is still a competitor for (MS) and that its energy 
is strictly smaller than the energy of $\hat{v}$, which contradicts the minimality.
Hence $\hat{v}$ is constant in $G$.
In view of this, we may also assume that if $x\in \Gamma\setminus \partial^* G$, 
then $x$ is not a point of density $1$ for $G$. Otherwise we would get $[\hat{v}](x) = 0$, 
where $[\hat{v}](x)$ denotes the difference of the traces of $\hat{v}$ at $x$.

Let us divide $G$ in the union of its \textit{indecomposable components} according to 
\cite[Theorem 1]{ACMM}, i.e., let $(G_i)_{i\in\mathbb{N}}$ be a family of sets with finite 
perimeter such that $G = \cup_{i\in\mathbb{N}}G_i$, $\mathcal{H}^1(\partial G) = 
\sum_{i\in\mathbb{N}}\mathcal{H}^1(\partial G_i)$, $\mathcal{L}^2(G_h \cap G_k) = 0$, 
$\mathcal{H}^1(\partial^*G_h \cap \partial^*G_k) = 0$ for every $h\neq k$, and such that
for every $k \in \mathbb{N}$ the set $G_k$ cannot be written as $G_k = G_k^1 \cup G_k^2$ 
with $\mathcal{L}^2(G^1_k \cap G^2_k) = 0$ and $\mathcal{H}^1(\partial^*G_k) = \mathcal{H}^1(\partial^*G^1_k) + \mathcal{H}^1(\partial^*G^2_k)$.

Let us set 
\begin{equation*}
 \Gamma_*:= \partial^* G = \bigcup_{j = 0}^{\infty} \partial^* G_j.
\end{equation*}

\textit{Choice of minimizers for (MS).} Let us choose the minimizer $\hat{v}$ by requiring
\begin{equation}\label{choiv}
\mbox{ess-}\inf_{\partial^* G_j} \hat{v}^+ \leq \hat{v}_{|G_j} \leq \mbox{ess-}\sup_{\partial^* G_j}\hat{v}^+,
\end{equation} 
where $\hat{v}^+$ denotes the trace of $\hat{v}$ external to $G_j$. 
In this way we have imposed a constraint on the constant values of $\hat{v}$ in the connected 
components of $Q_{\hat{\delta}}$ that do not touch $\partial Q_{\hat{\delta}}$. 

\textit{Comparison between $\hat{v}$ and $\tilde{v}$.}
We now prove (\ref{trick}). First of all we have that 
\begin{align}\label{EL}
\int_{Q_{\hat{\delta}}} \big(|\nabla  \tilde{v}|^2 - |\nabla  \hat{v}|^2\big)\,dx =& 
\int_{Q_{\hat{\delta}}} (\nabla  \tilde{v} - \nabla  \hat{v})\,(\nabla  \tilde{v} + \nabla  \hat{v})\,dx \nonumber\\
=& \int_{Q_{\hat{\delta}}} (\nabla  \tilde{v} - \nabla  \hat{v})\,\nabla \tilde{v}\,dx.
\end{align}
The last equality follows from 
\begin{equation*}
\int_{Q_{\hat{\delta}}} (\nabla \tilde{v} - \nabla  \hat{v})\, \nabla  \hat{v}\,dx = 0,
\end{equation*}
that is the Euler-Lagrange equation satisfied by $\hat{v}$, using as test function $\tilde{v} - \hat{v}$. 
Integrating by parts (\ref{EL}) we get
\begin{align}\label{harm}
\int_{Q_{\hat{\delta}}} \big(|\nabla  \tilde{v}|^2 - |\nabla  \hat{v}|^2\big)\,dx =& 
- \int_{Q_{\hat{\delta}}} (\tilde{v} - \hat{v})\,\Delta \tilde{v}\,dx 
+ \, \int_{\partial Q_{\hat{\delta}}} (\tilde{v} - \hat{v})\,\frac{\partial \tilde{v}}{\partial \nu}\, d\mathcal{H}^{1}   \nonumber\\
-&\, \int_{S_{\hat{v}}} \frac{\partial \tilde{v}}{\partial \nu} \, [\hat{v}]\, d\mathcal{H}^{1}. 
\end{align}
Notice that in the right-hand side of (\ref{harm}) 
the first two terms vanish because $\tilde{v}$ is harmonic and $\hat{v}=\tilde{v}$ on $\partial Q_{\hat{\delta}}$. 
Therefore, (\ref{harm}) reduces to
\begin{equation}\label{starina}
\int_{Q_{\hat{\delta}}} \big(|\nabla  \tilde{v}|^2 - |\nabla  \hat{v}|^2\big)\,dx = - 
\int_{S_{\hat{v}}} \frac{\partial \tilde{v}}{\partial \nu} \, [\hat{v}]\, d\mathcal{H}^{1}.
\end{equation}
We want now to give an estimate of the last term in the previous expression.
For the normal derivative of $\tilde{v}$, using the harmonicity of $\tilde{v}$ we get
\begin{equation}\label{bou2}
\Big|\frac{\partial \tilde{v}}{\partial \nu}\Big| \leq  \sup_{Q_\delta} |\nabla  \tilde{v}| 
\leq C(\delta, \hat{\delta})\, ||\nabla  \tilde{v} ||_{L^2(Q_{\hat{\delta}})}.
\end{equation}
It remains to estimate $\int_{S_{\hat{v}}}|[\hat{v}]|\, d\mathcal{H}^{1}$.

\textit{Estimate for the jump of $\hat{v}$.} 
Let us fix $x\in S_{\hat{v}}$ and let us define the set 
\begin{equation*}
 C(x):= \big\{r \in [0, 2\,\beta]: \partial B_r(x)\cap S_{\hat{v}} = \emptyset \big\}.
\end{equation*}
As $\mathcal{H}^1(S_{\hat{v}}) < \beta$, we conclude that 
\begin{equation*}
\mathcal{H}^1(C(x)) \geq \beta  
\end{equation*}
and this estimate holds true for every $x\in S_{\hat{v}}$.

Let us now take $r \in C(x)$, $\xi, \zeta \in \partial B_{r}(x)$. Let us consider the angles 
$\varphi,\psi \in [0, 2\pi)$ such that
\begin{equation*}
 \xi = x + (r\,\cos \varphi, r\,\sin \varphi), \quad  \zeta = x + (r\,\cos \psi, r\,\sin \psi),
\end{equation*}
and assume for instance that $\psi < \varphi$. Then we can write
\begin{align*}
 |\hat{v}(\xi) - \hat{v}(\zeta)| = \Big|\int_{\psi}^{\varphi}\partial_{\vartheta}
 \hat{v}(r,\vartheta)\,d \vartheta \Big|
\leq \sqrt{\varphi - \psi}\, \Big(\int_{\psi}^{\varphi} |\partial_{\vartheta}
 \hat{v}(r,\vartheta)|^2\,d \vartheta\Big)^{1/2}.
\end{align*}
Using the fact that $\partial_{\vartheta} = -r \sin \vartheta \partial_1 + r \cos \vartheta \partial_2$  
and the bound $(\varphi - \psi) < 2 \pi$, we have
\begin{equation*}
|\hat{v}(\xi) - \hat{v}(\zeta)| \leq c\,\Big(\int_{\psi}^{\varphi} r^2 |\nabla \hat{v}|^2 d \vartheta\Big)^{1/2} 
\leq  c\,\Big(\int_{0}^{2\,\pi} r^2 |\nabla \hat{v}|^2 d \vartheta\Big)^{1/2}.
\end{equation*}
Hence, since the previous estimate holds true for every $\xi,\zeta \in\partial B_r(x)$, we have
\begin{equation}\label{MPP}
\frac{1}{\sqrt{r}}\,\sup_{\xi,\zeta \in\partial B_r(x)}|\hat{v}(\xi) - \hat{v}(\zeta)| 
\leq c\,\Big(\int_{0}^{2\,\pi} r |\nabla \hat{v}|^2 d \vartheta\Big)^{1/2}.
\end{equation}

\textit{Maximum principle.} For every $x \in S_{\hat{v}}$ and for a.e. $r \in C(x)$ we have
\begin{equation}\label{Mp}
|[\hat{v}](x)| \leq \sup_{\xi,\zeta \in\partial B_r(x)}|\hat{v}(\xi) - \hat{v}(\zeta)|.
\end{equation}
Indeed, we can define the new function 
\begin{equation*}
\hat{v}_r:= 
\begin{cases}
m_r \vee (M_r \wedge \hat{v}) & \hbox{in } B_r(x),\\
\hat{v} & \hbox{otherwise in } Q_{\hat{\delta}}, 
\end{cases}
\end{equation*}
where 
\begin{equation*}
m_r:= \min_{\partial B_{r}(x)}\hat{v} \quad \hbox{and} \quad M_r:= \max_{\partial B_{r}(x)}\hat{v}. 
\end{equation*}
The function $\hat{v}_r$ is still a competitor for the minimum of (MS) and it coincides with $\hat{v}$ 
by (\ref{choiv}). Hence either $\hat{v}_r = \hat{v}$, 
or the energy associated to $\hat{v}_r$ is greater or equal to the energy corresponding to $\hat{v}$. Since, by definition, the truncation reduces the energy, we conclude that 
$\hat{v}_r = \hat{v}$. 
This gives immediately that $\hat{v}$ satisfies the maximum principle in the ball $B_r(x)$, hence 
(\ref{Mp}) is satisfied.

From (\ref{MPP}) and (\ref{Mp}) we obtain the inequality
\begin{equation*}
\frac{1}{\sqrt{r}}\,|[\hat{v}](x)|
\leq c\,\Big(\int_{0}^{2\,\pi} r |\nabla \hat{v}|^2 d \vartheta\Big)^{1/2}.
\end{equation*}
Squaring and integrating over $C(x)$ yields
\begin{equation*}
|[\hat{v}](x)|^2 \int_{C(x)}\frac{1}{r}\, dr
\leq c\,\int_{C(x)}\int_{0}^{2\pi} |\nabla \hat{v}|^2 r\,dr\,d \vartheta. 
\end{equation*}
Since $C(x) \subset [0, 2\,\beta]$, we have
\begin{equation*}
\int_{C(x)}\frac{1}{r}\, dr \geq \frac{1}{2\,\beta}\, \mathcal{H}^1(C(x)) \geq \frac{1}{2}, 
\end{equation*}
hence we deduce
\begin{equation*}
|[\hat{v}](x)| \leq c\,\Big(\int_{B_{2\,\beta}(x)} |\nabla \hat{v}|^2 dz\Big)^{1/2}
\end{equation*}
for $\mathcal{H}^1$-a.e. $x \in S_{\hat{v}}$. Moreover, since $\beta<(\delta-\hat{\delta})/2$, we have that
$B_{2\,\beta}(x) \subset Q_{\hat{\delta}}$ for every $x \in S_{\hat{v}}$, so that
\begin{equation*}
|[\hat{v}](x)| \leq c\,\Big(\int_{Q_{\hat{\delta}}} |\nabla \hat{v}|^2 dz\Big)^{1/2}.
\end{equation*} 
By integrating the previous expression over $S_{\hat{v}}$ we obtain
\begin{equation}\label{bou1}
\int_{S_{\hat{v}}}|[\hat{v}]|\,d\mathcal{H}^1 \leq c\,\mathcal{H}^1(S_{\hat{v}})||\nabla \hat{v}||_{L^2(Q_{\hat{\delta}})}.
\end{equation} 
Combining together (\ref{starina}), (\ref{bou2}) and (\ref{bou1}) we obtain
\begin{equation}\label{harm2}
\int_{Q_{\hat{\delta}}} \big(|\nabla \tilde{v}|^2 - |\nabla \hat{v}|^2\big)\,dx \leq 
2\,c\,C(\delta,\hat{\delta})\,\mathcal{H}^{1}(S_{\hat{v}})\,||\nabla \tilde{v} ||_{L^2(Q_{\hat{\delta}})} ||\nabla \hat{v}||_{L^2(Q_{\hat{\delta}})}. 
\end{equation}
Using in (\ref{harm2}) the Young inequality $2\,ab \leq a^2 + b^2$, which holds true for every $a,b >0$, we have
\begin{equation*}
\int_{Q_{\hat{\delta}}} \big(|\nabla  \tilde{v}|^2 - |\nabla  \hat{v}|^2\big)\,dx \leq c\,C(\delta,\hat{\delta}) \mathcal{H}^{1}(S_{\hat{v}})\big(||\nabla  \tilde{v} ||^2_{L^2(Q_{\hat{\delta}})} + ||\nabla \hat{v}||^2_{L^2(Q_{\hat{\delta}})}\big).
\end{equation*}
Being $\mathcal{H}^{1}(S_{\hat{v}}) < \beta$, we finally have
\begin{equation}\label{diseg}
\int_{Q_{\hat{\delta}}} |\nabla  \hat{v}|^2\,dx \geq \Big(\frac{1 - c\,\beta}{1 + c\,\beta}\Big)
\int_{Q_{\hat{\delta}}} |\nabla  \tilde{v}|^2\,dx,
\end{equation}
where $c>0$ is a constant depending only on the geometry of the problem. 
The estimate (\ref{diseg}) gives (\ref{trick}) with $\omega(\beta):= 2c\beta/(1 + c\beta)$.
\end{proof}

\bigskip
\bigskip
\centerline{\textsc{Acknowledgments}}
\bigskip
\noindent
I warmly thank Gianni Dal Maso for having proposed to me the study of this problem and 
for many interesting discussions, and Maria Giovanna Mora for several stimulating suggestions.
I would like to thank also Massimiliano Morini for his valuable comments on the subject of this paper.

\end{document}